\documentclass[10pt]{article}
\usepackage[utf8]{inputenc}
\usepackage[T1]{fontenc}
\usepackage[english]{babel}
\usepackage{amsmath,amssymb,amsthm}
\usepackage{graphicx} 
\usepackage[margin=1.5in]{geometry}
\usepackage[shortlabels]{enumitem}
\usepackage{todonotes}
\usepackage[bibstyle=numeric, backend=biber, sorting=none, citestyle=numeric-comp, giveninits,url=false]{biblatex} 
\usepackage{csquotes}\MakeOuterQuote"
\usepackage[colorlinks,citecolor=blue,linkcolor=blue]{hyperref}
\usepackage[capitalize]{cleveref}
\usepackage{comment}

\newtheorem{theorem}{Theorem}
\newtheorem{introtheorem}{Theorem}

\newtheorem{lemma}[theorem]{Lemma}
\newtheorem{corollary}[theorem]{Corollary}
\newtheorem{introcorollary}[introtheorem]{Corollary}
\newtheorem{introconjecture}[introtheorem]{Conjecture}
\newtheorem{proposition}[theorem]{Proposition}

\newtheorem*{problem*}{Problem}
\newtheorem{definition}[theorem]{Definition}
\newtheorem{conjecture}[theorem]{Conjecture}
\newtheorem*{assumption*}{Assumption}
\newtheorem{example}[theorem]{Example}
\theoremstyle{definition}
\newtheorem{remark}[theorem]{Remark}

\newcommand\up[1]{^{(#1)}}

\newcommand{\mc}{\mathcal}

\DeclareMathOperator{\id}{id}
\renewcommand{\H}{\mc H}
\newcommand{\A}{\mc A}
\newcommand{\K}{\mc K}
\newcommand{\B}{\mc B}
\newcommand{\J}{\mc J}
\newcommand{\ox}{\otimes}
\newcommand{\states}{\mathfrak S}

\newcommand{\T}{\mc T}
\newcommand{\NN}{\mathbb N}
\newcommand{\RR}{\mathbb R}
\newcommand{\DD}{\mathbb D}
\newcommand\M{\mc M}
\newcommand{\CC}{\mathbb{C}}
\newcommand{\tr}{{\mathrm{tr}}}
\newcommand\1{\mathbf1}
\newcommand{\ip}[2]{\langle #1,#2\rangle}

\newcommand{\ketbra}[2]{|#1\rangle\langle#2|}
\newcommand{\ket}[1]{|#1\rangle}
\newcommand{\kettbra}[1]{\ketbra{#1}{#1}}
\newcommand{\bra}[1]{\langle#1|}
\newcommand{\proj}[1]{\ketbra{#1}{#1}}
\newcommand{\oo}{\infty}
\newcommand{\placeholder}{\,\cdot\,}
\newcommand{\norm}[1]{\left\| #1 \right\|}

\newcommand{\supp}{\mathrm{supp}}
\newcommand\hide[1]{{}}

\renewcommand{\Re}{\mathrm{Re}}

\newcommand{\petz}[1]{D_{#1}^{\textup{P}}}
\newcommand{\sand}[1]{D_{#1}^{\min}}
\newcommand{\geom}[1]{D_{#1}^{\max}}
\newcommand{\meas}[1]{D_{#1}^{\textup{\textup{meas}}}}

\newcommand\undermat[2]{%
  \makebox[0pt][l]{$\smash{\underbrace{\phantom{%
    \begin{matrix}#2\end{matrix}}}_{\text{$#1$}}}$}#2}

\title{Sufficiency of R\'enyi divergences}
\author{Niklas Galke$^1$\footnote{niklas.galke@uab.cat} , Lauritz van Luijk$^2$\footnote{lauritz.vanluijk@itp.uni-hannover.de} , Henrik Wilming$^2$\footnote{henrik.wilming@itp.uni-hannover.de}}

\date{%
    $^1$Física Teòrica: Informació i Fenòmens Quàntics, Departament de Física, Universitat
    Autònoma de Barcelona, 08193 Bellaterra, Spain\\%
    $^2$Leibniz Universit\"at Hannover, Appelstraße 2, 30167 Hannover, Germany\\[2ex]%
    \today
}
\begin{document}

\maketitle
\vspace{-20pt}
\begin{abstract}
A set of classical or quantum states is equivalent to another one if there exists a pair of classical or quantum channels mapping either set to the other one. 
For dichotomies (pairs of states), this is closely connected to (classical or quantum) R\'enyi divergences (RD) and the data-processing inequality: If a RD remains unchanged when a channel is applied to the dichotomy, then there is a recovery channel mapping the image back to the initial dichotomy.
Here, we prove for classical dichotomies that equality of the RDs alone is already sufficient for the existence of a channel in any of the two directions and discuss some applications.  
In the quantum case, all families of quantum RDs are seen to be insufficient because they cannot detect anti-unitary transformations. Thus, including anti-unitaries, we pose the problem of finding a sufficient family. It is shown that the Petz and maximal quantum RD are still insufficient in this more general sense and we provide evidence for sufficiency of the minimal quantum RD.
As a side result of our techniques, we obtain an infinite list of inequalities fulfilled by the classical, the Petz quantum, and the maximal quantum RDs. These inequalities are not true for the minimal quantum RDs. 
Our results further imply that any sufficient set of conditions for state transitions in the resource theory of athermality must be able to detect time-reversal.
\end{abstract}

\tableofcontents

\section{Introduction}

A statistical experiment is generally described by an indexed set of probability measures or quantum states $(\rho_\theta)_{\theta \in \Theta}$ (we use the notation for density matrices throughout this introduction and use the word state in both cases). 
The parameter $\theta$ can be discrete or continuous and may be thought of as a parameter that controls the experiment and is sought to be estimated. 
For example, in the special case of \emph{binary} experiments with $\Theta=\{0,1\}$ it may be thought of as labeling two hypotheses to be distinguished. 
One experiment $(\rho_\theta\up1)_{\theta\in\Theta}$ can be considered as at least as informative as another one $(\rho_\theta\up2)_{\theta\in\Theta}$ if there exists a (classical or quantum) channel $T$ so that $T(\rho_\theta\up 1) = \rho_\theta\up 2$. In this case we write $(\rho_\theta\up1)_{\theta\in\Theta} \rightarrow (\rho_\theta\up2)_{\theta\in\Theta}$.
The channel $T$ may be viewed as a randomization or coarse-graining operation. 
It is a classic problem in statistics to provide necessary and sufficient conditions that characterize when one experiment is more informative than another one \cite{torgersen_comparison_1991}.

In this paper, we consider a much simpler problem: When is it the case that two experiments are \emph{equivalent}, meaning that either experiment is at least as informative as the other? 
In other words, we ask when there exist (classical or quantum) channels $T$ and $R$ such that\footnote{In the classical case, a channel is a stochastic map. In the quantum case, it is a completely positive trace-preserving map.}
\begin{align}
    T(\rho_\theta\up 1) = \rho_\theta\up 2,\quad R(\rho_\theta\up 2)=\rho_\theta\up 1\quad\forall \theta \in \Theta.
\end{align}

If this is the case, we will also say that the two sets are \emph{interconvertible} and write
\begin{align}
    (\rho_\theta\up1)_{\theta\in\Theta} \longleftrightarrow(\rho_\theta\up2)_{\theta\in\Theta}.
\end{align}
In the language of statistics, interconvertible sets of states describe experiments that are \emph{sufficient} with respect to each other \cite{torgersen_comparison_1991,petz_sufficient_1986,petz_sufficiency_1988,Jencova2006,PetzReview}. Our aim is to understand the structure of such states and relate  interconvertibility of binary experiments with equality of R\'enyi divergences. We will see that in the classical case, the R\'enyi divergences precisely contain the information that is needed to determine interconvertibility and conjecture that the same holds for suitable quantum R\'enyi divergences.

We first consider the case of density matrices $\rho_\theta\up j$ (finite or infinite-dimensional) and discuss a minimal normal form for two interconvertible sets of density matrices and the associated quantum channels $T,R$ (Section~\ref{sec:structure}). 
The normal form is based on the Koashi-Imoto theorem \cite{KoashiImoto} and its infinite-dimensional generalizations \cite{Jencova2006,Kuramochi2018}.
From the normal form it follows that (\cref{thm:noteleportation}):
\begin{enumerate}
    \item Non-commuting sets of density matrices cannot be interconverted with commuting sets of density matrices, a result closely connected to the "no-teleportation theorem". 
    \item The set of fixed points of any entanglement-breaking channel is commutative.
\end{enumerate}

Having established the normal form for interconvertible density matrices, we then specialize to the case of binary experiments in \cref{sec:dichotomies}.
This case simply corresponds to pairs of states also known as \emph{dichotomies} $(\rho,\sigma)$. We also say that a dichotomy $(\rho_1,\sigma_1)$ \emph{is  convertible to} or that it \emph{(quantum) relatively majorizes} $(\rho_2,\sigma_2)$ if $(\rho_1,\sigma_1) \rightarrow (\rho_2,\sigma_2)$ \cite{hardy1952inequalities,blackwell_equivalent_1953,ruch_principle_1976,ruch_mixing_1978,marshall_inequalities_2011,horodecki_fundamental_2013,horodecki_quantumness_2013,renes_relative_2016,gour_quantum_2018,wang_resource_2019,buscemi_information-theoretic_2019}. (To ease notation, we index different dichotomies by subscripts instead of superscripts.)
The study of dichotomies and their convertibility is closely related to the study of \emph{divergences} of states, which are functions on dichotomies that characterize how distinguishable the two states are. 
The central properties of a general divergence $\DD$ are:
\begin{enumerate}
    \item Positivity: $\DD(\rho,\sigma)\geq 0$ with equality if and only if $\rho=\sigma$,
    \item Data-processing: $\DD\big(T(\rho),T(\sigma)\big)\leq \DD(\rho,\sigma)$.
\end{enumerate}
Many divergences, in particular the R\'enyi divergences we will see later, also satisfy the additional property of being additive over tensor-products:
\begin{align}
\DD(\rho_1\otimes\rho_2,\sigma_1\otimes\sigma_2) = \DD(\rho_1,\sigma_1)+\DD(\rho_2,\sigma_2).
\end{align}
The prototypical example  of an additive divergence is the quantum relative entropy
\begin{align}
    D(\rho,\sigma) = \tr\big[\rho\big(\log(\rho)-\log(\sigma)\big)\big],
\end{align}
with $D(\rho,\sigma) = \infty$ if the support of $\rho$ is not contained in that of $\sigma$.
However, divergences come in many flavors, and we will see a variety below. 
One way to give operational meaning to them is in asymptotic settings, where one asks for the maximum rate $R=m/n$ to convert $(\rho_1,\sigma_1)^{\otimes n}$ into $(\rho_2,\sigma_2)^{\otimes m}$ (possibly up to a small error). 
It follows immediately from these properties that if $(\rho_1,\sigma_1)^{\otimes n}$ can be converted into $(\rho_2,\sigma_2)^{\otimes n}$, then $R\leq \DD(\rho_1,\sigma_1)/\DD(\rho_2,\sigma_2)$ for any valid additive divergence. 
In fact asymptotic conversion rates can often be expressed exactly as ratios of suitable divergences of the dichotomies \cite{matsumoto_reverse_2010,wang_resource_2019,buscemi_information-theoretic_2019,farooq_asymptotic_2023,lipka-bartosik_quantum_2023}. 

Here, we study the role of divergences for the interconvertibility of dichotomies without taking an asymptotic limit. 
It is clear from the data-processing inequality that any valid divergence must be invariant among interconvertible dichotomies:
\begin{align}\label{eq:interconvertibiblity-to-divergence}
    (\rho_1,\sigma_1) \longleftrightarrow (\rho_2,\sigma_2) \ \Rightarrow\ \DD(\rho_1,\sigma_1)=\DD(\rho_2,\sigma_2).
\end{align}
Moreover, it was shown in \cite{petz_sufficient_1986,petz_sufficiency_1988,Jencova2006} that if $(\rho_1,\sigma_1)\rightarrow (\rho_2,\sigma_2)$ \emph{and} it is true that $D(\rho_1,\sigma_1) = D(\rho_2,\sigma_2)$, then the two dichotomies are automatically also interconvertible. 
The same statement holds true for classical probability measures when $D$ is replaced by the Kullback-Leibler (KL) divergence \cite{kullback1951information}. It also holds when the relative entropy (or KL divergence) is replaced by a suitable \emph{$f$-divergence} \cite{hiai} (which we discuss in detail in \cref{sec:fdiv}).

But what happens if we do not know whether one of two divergences is convertible to the other one? Could it be true that equality of a family of divergences alone already implies that two dichotomies are interconvertible?
Let $\{\DD_\lambda\}_{\lambda \in \Lambda}$ be a family of divergences labelled by a parameter $\lambda$. 
We define such a family to be \emph{sufficient} if
\begin{align}
    \DD_\lambda(\rho_1,\sigma_1) = \DD_\lambda(\rho_2,\sigma_2)\ \ \forall \lambda\in\Lambda\  \Rightarrow\ (\rho_1,\sigma_1) \longleftrightarrow (\rho_2,\sigma_2). 
\end{align}
Thus equality of a sufficient family of divergence is equivalent to interconvertibility.
It is at first not clear that a sufficient family of divergences, in fact, exists, neither in the classical nor in the quantum case.
As a first result, we prove that the classical R\'enyi divergences $D_\alpha$ constitute a sufficient family of divergences (Section~\ref{sec:classical}):

\begin{introtheorem}[Sufficiency of R\'enyi divergences]\label{thm:classical_solution_intro}
Let $(X_i,\mu_i)$, $i=1,2$, be measure spaces and $p_i\ll q_i$ probability measures in $L^1(X_i,\mu_i)$.
    If there is an open interval $(a,b)\subset\RR$, $a<b$,  such that
    \begin{equation}
        D_\alpha(p_1,q_1) = D_\alpha(p_2,q_2) <\oo \quad\forall\alpha\in(a,b),
    \end{equation}
    there are stochastic maps between $L^1(X_1,\mu_1)$ and $L^1(X_2,\mu_2)$ interconverting the dichotomies $(p_1,q_1)$ and $(p_2,q_2)$.
\end{introtheorem}
    This theorem follows from the existence of a "normal form" for classical dichotomies, which can be determined from the R\'enyi divergences alone:

\begin{introtheorem}\label{introthm:ptildes}
    Let $p\ll q$ be probability distributions on a measurable space $X$.
    Then there is a probability measure $\tilde q$ on $(-\oo,\oo]$ such that 
    \begin{equation}\label{eq:intro_renyi_with_tilde}
        D_\alpha(p,q) = \frac1{\alpha-1}\log \int_{(-\oo,\oo]} e^{-t\alpha}d\tilde q(t)\quad \forall \alpha\in\RR.
    \end{equation}
    $\tilde p := e^{-t}\tilde q\ll \tilde q$ is a probability measure and  $(\tilde p,\tilde q)\leftrightarrow(p,q)$.
    If the R\'enyi divergences $D_\alpha(p,q)$ exist for an open set of $\alpha\in\RR$, then $\tilde q$ and, hence, $\tilde p$ are uniquely determined by \eqref{eq:intro_renyi_with_tilde}.
\end{introtheorem}

Since R\'enyi divergences are additive, one immediate consequence of Theorem~\ref{thm:classical_solution_intro} is that there is no non-trivial \emph{(strict) catalysis} for interconversion of dichotomies (see Refs.~\cite{Datta2022,lipka-bartosik_catalysis_2023} for reviews on catalysis): If $(p_1\otimes r,q_1 \otimes s)\leftrightarrow (p_2\otimes r, q_2\otimes s)$ with $r\ll s$ and the R\'enyi divergences for both are finite on an open set, then also $(p_1,q_1)\leftrightarrow (p_2, q_2)$.
This is noteworthy, because non-trivial cases of strict catalysis exist for one-sided conversion of dichotomies \cite{jonathan_entanglement-assisted_1999,brandao_second_2015,rethinasamy_relative_2020}: $(p_1\otimes r,q_1\otimes s) \rightarrow (p_2\otimes r,q_2 \otimes r)$ does \emph{not} imply $(p_1,q_1)\rightarrow (p_2,q_2)$.
In the context of resource theories, a family of divergences $\DD_\lambda$ is called a \emph{complete set of monotones} if
\begin{align}\label{eq:completeness}
    \DD_\lambda(\rho_1,\sigma_1) \geq \DD_\lambda(\rho_2,\sigma_2) \quad \forall \lambda\in\Lambda \ \Rightarrow (\rho_1,\sigma_1) \rightarrow (\rho_2,\sigma_2).
\end{align}
The existence of non-trivial examples of catalysis alone implies that no family of additive divergences (such as the R\'enyi divergences) can constitute a complete set of monotones \cite{fritz_resource_2017}.
This may lead to the impression that the R\'enyi divergences with $\alpha>0$ do not contain sufficient information to decide whether a classical dichotomy $(p_1,q_1)$ can be converted to another one $(p_2,q_2)$. 
Our results imply that this is false. 
Consider the dichotomies $(\tilde p_1,\tilde q_1)$ and $(\tilde p_2,\tilde q_2)$ whose existence is implied by \cref{introthm:ptildes} and which are completely determined by the R\'enyi divergences with $\alpha>0$. 
We have $(p_1,q_1) \rightarrow (p_2,q_2)$ \emph{if and only if} $(\tilde p_1,\tilde q_1) \rightarrow (\tilde p_2,\tilde q_2)$. Since the latter dichotomies are determined by the R\'enyi divergences, it follows that the divergences, in principle, contain sufficient information to decide convertibility:
\begin{introcorollary}
    It is in principle possible to decide if $(p_1,q_1)\rightarrow(p_2, q_2)$ from the knowledge of the R\'enyi divergences $D_\alpha( p_1, q_1)$ and $D_\alpha( p_2, q_2)$ on any non-empty open subset $(a,b)\subset [0,\oo)$ on which they are both finite.
\end{introcorollary}
In \cref{sec:explicit}, we discuss and illustrate these results in detail for the case of probability vectors, where every step can be carried out explicitly.

Theorem~\ref{thm:classical_solution_intro} relies on observations relating R\'enyi divergences to the theory of Laplace transforms, which we consider to be of independent interest and further explore in Section~\ref{sec:CM}.
This connection gives rise to an infinite set of inequalities between R\'enyi divergences, which also hold true for the \emph{Petz quantum R\'enyi divergence} $\petz\alpha$ as well as the \emph{maximal} (or \emph{geometric}) \emph{quantum R\'enyi divergence} $\geom\alpha$ in the quantum case. We refer to \cite{tomamichel} and \cref{sec:dichotomies} for an introduction to the various different quantum R\'enyi divergences. As an example of these inequalities, we find:
\begin{equation}\label{eq:inequalities-intro}
 (-1)^n\frac{d^n}{d\alpha^n}\Big(e^{(\alpha-1)(\DD_\alpha(\rho,\sigma)-\DD_\oo(\rho,\sigma))}\Big)\ge0 \quad \forall\alpha>0,\, n\in\NN,
\end{equation}
if $\DD_\alpha =\petz\alpha,\geom\alpha$ and in the classical case (always assuming that $\DD_\infty(\rho,\sigma)$ is finite). In fact, we show that \eqref{eq:inequalities-intro} is true for a sufficient family of R\'enyi divergences (should it exist) if and only if $[\rho,\sigma]=0$. In particular, neither the Petz- nor the maximal R\'enyi divergence are sufficient, since they can be "simulated" by classical distributions: For any dichotomy $(\rho,\sigma)$ there exists a classical dichotomy $(p,q)$, called \emph{Nussbaum-Szko\l a distribution} \cite{NussbaumSzkola}, such that $\petz\alpha(\rho,\sigma) = D_\alpha (p,q)$ (and similarly for $\geom\alpha$ for a suitable, different classical dichotomy).

Given the sufficiency of classical R\'enyi divergences, it is natural to conjecture that some quantum R\'enyi divergence is sufficient as well.\footnote{We conjectured this in a previous version of this paper.}
In \cref{sec:insufficiency_of_known} we see that all customary quantum R\'enyi divergences known to us are not sufficient.

\begin{proposition} The following quantum divergences are not sufficient: 
The standard quantum $f$-divergences -- including the Petz quantum Rényi divergences -- the maximal quantum $f$-divergences, the minimal (or sandwiched) R\'enyi divergences, the $\alpha$-$z$ R\'enyi divergences, the measured R\'enyi divergences, the sharp R\'enyi divergences, the Kringel R\'enyi divergences.  Moreover, the collection of all these divergences is also not sufficient.
\end{proposition}

For the convenience of the reader, in \cref{sec:zoo}, we provide definitions and summarize the main properties of all these divergences. 
The reason why none of these divergences are sufficient is very simple: Each of them is not only invariant under unitary conjugation but also invariant under anti-unitary conjugation. 
In particular, they assign the same values to $(\rho,\sigma)$ and $(\overline{\rho},\sigma)$, where $\overline\rho$ denotes entry-wise complex conjugate of $\rho$ in the eigenbasis of $\sigma$. However, $(\rho,\sigma)$ in general cannot be interconverted with $(\overline{\rho},\sigma)$ using completely positive trace-preserving maps. This can be shown using the normal form of interconvertible dichotomies.

This result shows that asking for suffiency for completely positive interconversion is too restrictive in the quantum case.
Since conjugation by anti-unitaries is still positive, it suggests to generalize to interconvertibility  by positive trace-preserving maps in order to obtain a well-posed problem. Indeed it is known that several of the above divergences are also monotone under this more general class of maps (at least for certain ranges of $\alpha$), see \cref{sec:zoo}. 
Adopting this generalized notion of interconvertibility we can still rule out the Petz and maximal quantum R\'enyi divergences as sufficient families.
The $\alpha$-$z$ R\'enyi divergences are only known to be monotonic under positive maps for a smaller range of $(\alpha,z)$ compared to completely positive maps. 
We conjecture:
\begin{introconjecture}\label{conj:introconj} The \emph{minimal quantum R\'enyi divergences} $D^{\min}_\alpha$ (\emph{sandwiched quantum R\'enyi divergences}) provide a sufficient family of divergences for interconversion of finite-dimensional density matrices via positive, trace-preserving maps.
\end{introconjecture}

\noindent A positive resolution to the conjecture
\begin{enumerate}[a)]
    \item implies that the minimal quantum R\'enyi divergences encode all relevant relative information between two quantum states (in the sense of interconvertibility via positive trace-preserving maps), 
    \item provides a new perspective into the problem of recovery of quantum operations.
\end{enumerate}
So far, not even the sufficiency of classical R\'enyi divergences seemed to be known. 
We hope that formulating and discussing the conjecture in detail motivates more researchers to study this interesting problem and make further progress towards its resolution.
In \cref{sec:conjecture} we substantiate our conjecture by proving it for the case that both $\rho_1$ and $\rho_2$ are pure states as well as in the special case where $\rho_1$ is pure and $[\rho_2,\sigma_2]=0$.
In the latter case, it then follows that $[\rho_1,\sigma_1]=0$ and $\rho_2$ is pure as well. 
As a side-result of general interest, we thereby show that there exists no general equivalent of the Nussbaum-Szko\l a distributions for the minimal quantum R\'enyi divergence:
\begin{introcorollary}
    The minimal quantum R\'enyi divergences can, in general, not be classically simulated: There exist non-commuting dichotomies $(\rho,\sigma)$, such that there is no classical dichotomy $(p,q)$ with $D^{\min}_\alpha(\rho,\sigma) = D_\alpha(p,q)$ for all $\alpha\in [\frac{1}{2},\infty)$. 
\end{introcorollary}
This statement does not yet rule out the existence of classical simulations for certain specific dichotomies. But if our conjecture is true, then the minimal quantum R\'enyi divergence cannot be classically simulated for any non-abelian dichotomy.  In particular, the latter would always violate at least one of the inequalities \eqref{eq:inequalities-intro}. Conversely, \cref{conj:introconj} can be falsified by presenting a single example of a non-abelian dichotomy whose minimal quantum R\'enyi divergence can be classically simulated.
As an illustration and further evidence for our conjecture, we present additional numerical examples of non-abelian mixed dichotomies for which \eqref{eq:inequalities-intro} indeed does not hold for the minimal quantum R\'enyi divergence and hence for which no classical simulation exists (see Fig.~\ref{fig:qubits}).

\textbf{Application to Thermodynamics:} Our result for classical states has a direct application in the context of resource theories \cite{chitambar_quantum_2019}: The so-called resource theory of athermality \cite{Janzing2000,horodecki_fundamental_2013,horodecki_quantumness_2013,brandao_resource_2013,brandao_second_2015,yunger_halpern_beyond_2016} deals with dichotomies where the second state $\sigma$ is given by a Gibbs state $\omega_H :=\exp(-\beta H)/Z$ of some Hamiltonian $H$ with fixed inverse temperature $\beta> 0$ and all channels (free operations) must map Gibbs states to Gibbs states. 
In the case of abelian dichotomies $(\rho,\omega_H)$, the R\'enyi divergences $D_\alpha(\rho,\omega_H)$ with $\alpha\geq 0$ can be interpreted as generalized free energies $F_\alpha(\rho,H)$ of $\rho$ with respect to $H$. 
Our results then immediately imply that if $F_\alpha(\rho_1,H_1) = F_\alpha(\rho_2,H_2)$ for all $\alpha\geq 0$, then the two states are thermodynamically completely equivalent because they can be reversibly interconverted by thermodynamically free operations already in the single shot. Moreover, knowledge of the free energies for $\alpha>0$ is sufficient to decide whether a thermodynamic state-transition is possible.

The situation is quite different in the quantum setting.
Since no quantum R\'enyi divergence is sufficient for interconvertibility with quantum channels, equality of free energies (defined via quantum R\'enyi divergences with respect to the Gibbs states) does not imply equivalence of states in the resource theory of athermality. 
Interestingly, this is true because time-reversal, which is precisely an anti-unitary symmetry commuting with the Gibbs state, cannot be detected using free energies while it also cannot be implemented via completely positive maps. 
It is a longstanding an open problem to determine a meaningful and computable sufficient set of conditions for state transitions in the resource theory of athermality (see, however, Ref.~\cite{gour_quantum_2018}). Our result shows that such a set of conditions must be able to detect time-reversal. 

\paragraph{Structure of the paper:} We first discuss the general structure of interconvertible sets of quantum states in Sec.~\ref{sec:structure} and then specialize to dichotomies and the associated R\'enyi divergences in Sec.~\ref{sec:dichotomies}. In Sec.~\ref{sec:classical}, we prove the solution to the classical case. Starting from Sec.~\ref{sec:insufficiency_of_known}, we discuss the quantum case, which includes a discussion of the role of "complete monotonicity" leading to the inequalities \eqref{eq:inequalities-intro} in Sec.~\ref{sec:CM}.
In the Appendix, we discuss the proof of the classical conjecture without the assumption of absolute continuity (\ref{sec:no-abs-continuity}) and examine in detail the case of probability vectors (\ref{sec:explicit}). In particular, we also show that the question of convertibility of two dichotomies can be decided from the knowledge of the R\'enyi divergences for $\alpha$ in an open subset of $[0,\infty)$ alone.
\cref{sec:appendix} contains the proof of the connection between the Koashi-Imoto decomposition and the problem of interconvertibility in infinite dimensions.
\cref{sec:fdiv} contains a discussion of $f$-divergences.
Finally, in \cref{sec:zoo} we list all families of quantum Rényi divergences known to us and collect their relevant properties.

\section{Structure of interconvertible states}\label{sec:structure}

We will allow for possibly continuous classical systems and infinite-dimensional quantum systems.
For a quantum system $S$ we denote the Hilbert space by $\H$ and the set of density operators on $\H$ by $\states(S)$.
A classical system $S$ is described by a $\sigma$-finite measure space $(X,\mu)$. The state space is the set of probability distributions on $X$ which are absolutely continuous with respect to the reference measure $\mu$, i.e., $\states(S)$ is the space of positive normalized elements of $L^1(X,\mu)$.
We will usually identify an absolutely continuous measure with its corresponding $\mu$-density and write $L^1(X,\mu)$ for both sets.

A \emph{channel} between two systems $S_i$, $i=1,2$, is a linear\footnote{To be precise, we should say that $T$ is affine, i.e., $T(p\rho+(1-p)\sigma) = pT(\rho) +(1-p)T(\sigma)$, $0<p<1$, $\rho,\sigma\in\states(S_1)$. This, however, ensures that $T$ can be extended  \emph{linearly} to the whole span of the states.} completely positive map $T : \states(S_1)\to \states(S_2)$.
For quantum systems, this means that $T$ is a completely positive trace-preserving map $\T(\H_{S_1})\to\T(\H_{S_2})$ between the Banach spaces of trace class operators.
For classical systems, the definition is equivalent to that of a stochastic map.
If one of the systems is quantum and the other one is classical, channels describe either measurements or preparations depending on which system is classical or quantum.

\begin{definition}\label{def:interconvertibility}
    Consider two systems $S_1$ and $S_2$ and two families of states $(\rho_\theta\up i)_{\theta\in\Theta} \subset \states(S_i)$, $i=1,2$, indexed by a set $\Theta$.
    We say that $(\rho_\theta\up1)_{\theta\in\Theta}$ and $(\rho_\theta\up2)_{\theta\in\Theta}$ are interconvertible and write
    \begin{equation}\label{eq:interconvertibility}
        (\rho_\theta\up1)_{\theta\in\Theta} \longleftrightarrow(\rho_\theta\up2)_{\theta\in\Theta},
    \end{equation}
    if there are channels $T:\states(S_1)\to \states(S_2)$ and $R:\states(S_2)\to\states(S_1)$ so that $T(\rho_\theta\up1) = \rho_\theta\up2$, $R(\rho_\theta\up2)=\rho_\theta\up1$ for all $\theta\in\Theta$.
\end{definition}

Note that the systems in \cref{def:interconvertibility} can be classical or quantum systems.
The main result of this section is a structure theorem for interconvertibility.
To state it, we need the following result known as the Koashi-Imoto Theorem (see \cite{KoashiImoto}):

\begin{theorem}[Koashi-Imoto]\label{thm:KI}
Let $(\rho_\theta)_{\theta \in \Theta}$ be a collection of states on a separable Hilbert space $\mc H$ such that the union of the supports is (dense in) $\H$. 
Then $\mc H$ decomposes as
$\mc H = \bigoplus_{j=1}^N \mc J_j \otimes \mc K_j$, $N\in \mathbb N\cup\{\infty\}$, such that:
\begin{itemize}
	\item The states $\rho_\theta$ decompose as
	\begin{align}
		\rho_\theta = \oplus_jp_{j|\theta}\,\rho_{j|\theta} \otimes \omega_j,
	\end{align}
		where $\{p_{j|\theta}\}_j$ is a probability distribution, and $\rho_{j|\theta}$ is a density matrix for all $\theta\in\Theta$.
        Furthermore, $\omega_j$ are density matrices independent of $\theta$.
	\item For every channel $T$ on $\H$ that leaves the $\rho_\theta$ invariant we have
	\begin{align}\label{eq:KI_channel_form}
		T|_{\T (\mc J_j\otimes \mc K_j)} = \id\otimes T_j, 
	\end{align}
	with $\id$ on $\mc J_j$ and $T_j$ a channel on $\mc K_j$ such that $T_j(\omega_j)=\omega_j$. 
\end{itemize}
\end{theorem}

The original theorem was proved in the finite-dimensional case in \cite{KoashiImoto}.
An elegant algebraic proof is given in \cite{Hayden2004}.
A proof of the infinite-dimensional case can be found in \cite[Sec.~3.3.]{Kuramochi2018} (in that reference \cref{eq:KI_channel_form} is not proved, but it follows by the same argument as in the finite-dimensional case).

The family of states $(\rho_\theta)$ is interconvertible with the family of states $(\hat\rho_\theta)$ which is obtained from the Koashi-Imoto decomposition by $\hat \rho_\theta = \oplus_j p_{j|\theta}\rho_{j|\theta}$ on $\hat\H = \bigoplus_j\J_j$.
We will call the family of states $(\hat\rho_\theta)$ the {\bf Koashi-Imoto minimal form}.
Explicit channels interconverting $(\rho_\theta)$ with its Koashi-Imoto minimal form are 
\begin{equation}\label{eq:normal_intercon}
    \iota(\hat\sigma) = \oplus_j (P_{\J_j}\hat\sigma P_{\J_j}) \ox\omega_j, \quad\text{and}\quad \pi(\sigma) = \oplus_j \tr_{\K_j}[P_{\K_j\ox\J_j}\sigma P_{\K_j\ox\J_j}].
\end{equation}
Here, $P_V$ denotes the orthogonal projection onto the subspace $V$.
To be specific we have $\iota(\hat\rho_\theta) = \rho_\theta$ and $\pi(\rho_\theta) = \hat\rho_\theta$.

The following structure theorem says that two families of states are interconvertible if and only if they have unitarily equivalent Koashi-Imoto minimal forms.

\begin{theorem}\label{thm:interconvertibility}
    Let $(\rho_\theta\up i)_{\theta \in \Theta}$ be families of quantum states on quantum systems $S_i$, such that the union of the supports is all of $\H_i$, $i=1,2$.
    Consider the Koashi-Imoto decompositions
    \begin{align}
        \H_{S_i} &= \bigoplus_{j=1}^{N_i} \J_j\up i\otimes \K_j\up i,\quad \rho_\theta\up i = \oplus_{j=1}^{N_i}\,p_{j|\theta}\up i \ \rho_{j|\theta}\up i\ox \omega_j\up i, \quad \theta\in\Theta,\,i=1,2.
    \end{align}
    The following are equivalent
    \begin{enumerate}[(A)]
        \item\label{it:interconvertibility} $(\rho_\theta\up1)$ and $(\rho_\theta\up2)$ are interconvertible.
        \item\label{it:normal_form}
        There is a reordering of $j$-indices (and $N_1=N_2$) such that $p_{j|\theta}\up1=p_{j|\theta}\up2$ and there are unitaries $U_j:\J_j\up1\to\J_j\up2$ so that $\rho_{j|\theta}\up 1= U_j\rho_{j|\theta}\up 2U_j^*$.
    \end{enumerate}
    If the equivalent conditions hold and if $T$, $R$ are channels that interconvert the sets of states, then it follows that
    \begin{equation}\label{eq:intercon_channels}
        T |_{\T(\J_j\up1 \otimes\K_j\up1)} = (U_j (\cdot) U_j^*) \otimes T_j,
        \quad R |_{\T(\J_{j}\up2 \otimes\K_{j}\up2)} = (U_j^* (\cdot) U_j) \otimes R_j,
    \end{equation}
    where the $U_j$ are the unitaries from \ref{it:normal_form} and $T_j$ and $R_j$ are quantum channels with $T_j(\omega_j\up1) =\omega_j\up2$ and $R_j(\omega_j\up2) = \omega_j\up1$.
\end{theorem}

For finite-dimensional systems, this theorem is essentially a consequence of \cite[Thm.~6]{Jencova2006}. 
We include a detailed proof which also covers infinite-dimensional systems in \cref{sec:appendix}.
As an illustration, we consider the question of interconvertibility of two families of qubit states:

\begin{example}
    Consider two collections of qubit states $\rho_\theta\up i\!\in M_2(\CC)$, $i=1,2$.
    Then either the Koashi-Imoto decomposition of $\rho_\theta\up i$ has a single direct summand or all states commute.
    In particular, the two families are interconvertible if and only if they are unitarily equivalent or if all states in each family are identical.
    The second case occurs if and only if, for each family, the (necessarily single) space $\K$ is two-dimensional so that the $\J$ space is trivial.
\end{example}

One can see from \cref{thm:interconvertibility} that a set of states $(\rho_\theta)$ is interconvertible with a set of commuting quantum states if and only if the $\rho_\theta$ are jointly diagonalizable.
In fact, we can even rule out continuous classical systems:

\begin{proposition}\label{thm:noteleportation}
    Let $(\rho_\theta)$ be a collection of states on a quantum system $S$ (resp.\ finite-dimensional quantum system $S$). 
    The following are equivalent
    \begin{enumerate}[(1)]
        \item\label{it:pairw_comm} The states $\rho_\theta$ are pairwise commutative.
        \item\label{it:one_dim_KI} The Koashi-Imoto decomposition of $\H_S$ has one-dimensional $\J_j$ spaces (see Theorem \ref{thm:KI}).
        \item\label{it:intercon_finite_cl} $(\rho_\theta)$ is interconvertible with a set of commuting states on a quantum system (resp.\ finite-dimensional quantum system) $S$.
        \item\label{it:intercon_cont_cl} $(\rho_\theta)$ is interconvertible with a set of states on a (not necessarily discrete) classical system.
        \item\label{it:entanglement_breaking} There is an entanglement-breaking channel $T:\states(S)\to\states(S)$ that has every $\rho_\theta$ as a fixpoint.
    \end{enumerate}
\end{proposition}
\begin{remark}
   Point \ref{it:one_dim_KI} implies that for commuting states $(\rho_\theta)_{\theta\in\Theta}$ the Koashi-Imoto minimal form $\hat \rho_\theta$ of $\rho_\theta$ is simply given by the probability distribution $\{p_{j|\theta}\}_{j=1}^N$. See also \cref{sec:explicit} for an explicit discussion of the Koashi-Iomoto minimal form for finite-dimensional commuting states.
\end{remark}
\begin{proof}
    \ref{it:one_dim_KI} $\Rightarrow$ \ref{it:pairw_comm} is clear. The converse direction follows because any quantum channel on $\mc J_j$ that fixes all states $\rho_{j|\theta}$ (for fixed $j$) must be the identity channel per definition of the spaces $\mc J_j$. Since we assume that all the $\rho_{j|\theta}$ commute, this requires $\dim \mc J_j=1$.
    This and \cref{thm:interconvertibility} also implies the equivalence of \ref{it:pairw_comm} and \ref{it:one_dim_KI} with \ref{it:intercon_finite_cl}. The implication \ref{it:intercon_finite_cl} $\Rightarrow$ \ref{it:intercon_cont_cl} is clear.

    \ref{it:intercon_cont_cl} $\Leftrightarrow$ \ref{it:entanglement_breaking}:
    The assumption of \ref{it:intercon_cont_cl} guarantees that there is a channel $T : \states(S)\to\states(S)$ which has every $\rho_\theta$ as a fixed point and factorizes through a classical system.
    The result now follows from the fact channels factors through a classical system if and only if they are entanglement breaking.
    A proof for this which allows for continuous classical systems and infinite-dimensional quantum systems is given in \cite[Thm.~2]{HolevoEtAl2005} (their assumption that the classical system is a complete separable metric space is not used in the part of the argument proving this claim).

    \ref{it:intercon_cont_cl} $\Rightarrow$ \ref{it:one_dim_KI}:
    Let $(p_\theta)_{\theta\in\Theta}$ be a collection of states on a classical system described by $L^1(X,\mu)$ which are interconvertible with $(\rho_\theta)$.
    This means that the two von Neumann algebras with indexed families of states defined as $(\M_1,(\sigma_\theta\up1)) := (L^\infty(X,\mu),(p_\theta))$ and $(\M_2,(\sigma_\theta\up2)):=(\B(\H_S),(\rho_\theta)_{\theta\in\theta})$ are equivalent, i.e., there exist normal unital completely positive maps $\Phi:\M_1\to\M_2$ and $\Psi:\M_2\to\M_1$ such that $\sigma_\theta\up 1 \circ \Psi = \sigma_\theta\up2$ and $\sigma_\theta\up 2 \circ \Phi = \sigma_\theta\up1$ for all $\theta$.
    It follows from a recent result in the theory of sufficient subalgebras \cite[Thm.~1]{Kuramochi2017} that the minimal sufficient subalgebras of $\M_1$ and $\M_2$ are isomorphic.
    These subalgebras are
    \begin{equation}\label{eq:fixed_point_alg}
        \mc N_i = \bigcap_{\Phi\in \mathrm{Ch}(\M_i)} \bigl\{ x \in \M \,\big\vert\, \Phi(x)=x \ \text{ holds, if } \sigma_\theta\up i\circ\Phi=\sigma_\theta\up i \ \forall\theta\in\Theta \bigr\},
    \end{equation}
    where $\mathrm{Ch}(\M_i)$ is the set of normal unital completely positive maps $\M_i\to\M_i$.
    The restrictions of the $\sigma^{(2)}_\theta=\rho_\theta$ to $\mc N_2$ are just the minimal form obtained from the Koashi-Imoto decomposition of $\rho_\theta$, meaning that $\mc N_2 = \oplus_j \B(\J_j)$.
    On the other hand, $\mc N_1$ is commutative since it is a subalgebra of $L^\infty(X,\mu)$.
    Since $\mc N_1\cong \mc N_2$, all the $\J_j$ spaces must hence be one-dimensional.
\end{proof}

\cref{thm:noteleportation} can be interpreted as a generalization of the no-teleportation theorem as the latter is recovered if $(\rho_\theta)_{\theta\in\Theta}$ is the maximal family containing all states. 

\begin{remark}
    In the finite-dimensional case, there is a nice alternative proof of \ref{it:entanglement_breaking} $\Rightarrow$ \ref{it:intercon_finite_cl}:
    That $T$ is entanglement breaking implies that its Choi-Jamiol\-koswki matrix $\rho_T:= (\id\otimes T) (\proj{\Omega})$ is separable so that we can write $\rho_T = \sum_{i=1}^r p_i \,\omega_i\ox \tau_i$.    
    Since $T$ is trace-preserving we have $\sum_{i=1}^r p_i \omega_i = \1/d_S$. Defining the effect operators $E_i = p_i d_S \omega_i^\top$, where ${}^\top$ denotes the transposition induced by the choice of maximally entangled state $\ket\Omega$, we find that
    \begin{align}
        T(\sigma) = d_S\,\tr_1[(\sigma^\top\!\ox\1)\rho_T]=\sum_{i=1}^r  \tr[E_i \sigma]\, \tau_i.
    \end{align}
    Hence $T$ is a measure-and-prepare channel with respect to the POVM $\{E_i\}$ and thus factors through the classical system $X=\{1,\ldots,r\}$.
    In the infinite-dimensional case, a similar argument is possible. Here, one has to replace $\Omega$ by a vector state with a faithful marginal. However, this approach is complicated by the fact that separable states, in general, do not admit a countable convex combination into product states but might require an integral instead (see \cite{HolevoEtAl2005}).
\end{remark}

\begin{corollary}\label{thm:commnogo}
    Let $(\rho_\theta\up i)$ be interconvertible collections of states on systems $S_i$, $i=1,2$.
    If the $\rho_\theta\up 1$ commute (e.g., if $S_1$ is a classical system), the states $\rho_\theta\up2$ commute.
\end{corollary}

\begin{remark}\label{cor:positive-commnogo} As positive maps with abelian range or domain are completely positive, \cref{thm:commnogo} also applies to interconversion via positive instead of completely positive maps.
\end{remark}

    \begin{proof}
        Assume that $S_1$ is a quantum system.
        By \cref{thm:noteleportation}, commutativity of $\rho_\theta\up1$ implies that the spaces $\J_j\up i$ in \cref{thm:interconvertibility} are of dimension $1$.
        But said Theorem then implies that the same is true for $\rho_\theta\up2$.
        Applying \cref{thm:noteleportation} a second time shows the assertion.

        Now let $S_1$ be a classical system and let $T,R$ be interconverting channels. 
        Then $S = T\circ R:\states(S_2)\to\states(S_2)$ is a channel that factors through a classical system which has every $\rho_\theta\up2$ as a fixpoint. The result follows from \cref{thm:noteleportation}.
    \end{proof}
\begin{remark}
    For pairs $(\rho, \sigma)$ of finite-dimensional density matrices \cref{thm:commnogo} can also be deduced from \cite[Thm.~6]{berta_variational_2017}.
\end{remark}

\section{Dichotomies: Sufficient families of R\'enyi divergences}
\label{sec:dichotomies}

For this section, we restrict to interconvertibility in the case of $|\Theta|=2$.
That is we consider two \emph{dichotomies} $(\rho_i, \sigma_i)$, $i=1,2$, of states and ask whether there exist channels $T$ and $R$ mapping between the two systems such that
\begin{align*}
   &&T(\rho_1) &= \rho_2, & R(\rho_2)&= \rho_1,&&\\
    &&T(\sigma_1) &= \sigma_2, & R(\sigma_2)&=\sigma_1.&&
\end{align*}
We are interested in examining the connection of this problem to suitable families of (quantum) R\'enyi divergences.
For this reason, we will always make the following assumption when talking about dichotomies:
\begin{assumption*} 
In the main text, we only consider dichotomies $(\rho,\sigma)$ with $\rho\ll\sigma$.
For quantum states, this means that $\mathrm{supp}(\rho)\subseteq\mathrm{supp}(\sigma)$ while it means that $\rho$ is absolutely continuous with respect to $\sigma$ (in the sense of measures) for classical systems. The classical case is treated without the assumption on absolute continuity in \cref{sec:no-abs-continuity}.
\end{assumption*}

Here, a family of R\'enyi divergences means functions $\DD_\alpha(\rho,\sigma)$ defined for pairs of states with $\mathrm{supp}(\rho)\subseteq \mathrm{supp}(\sigma)$ and certain values of $\alpha$ and which satisfies certain axioms \cite[Sec.~4.2]{tomamichel}.
For us, the most important properties a family of R\'enyi divergences has to satisfy are the \emph{data processing inequality},
\begin{equation}
    \DD_\alpha(T(\rho),T(\sigma)) \le \DD_\alpha(\rho,\sigma),
\end{equation}
which should hold on a subset $\Lambda\subset\RR$ (independent of $(\rho,\sigma)$),
and the property that for commuting states one recovers the classical R\'enyi divergences, i.e., if $\rho = \sum p_i \ketbra ii$, $\sigma_i=\sum q_i\ketbra ii$ then $\DD_\alpha(\rho,\sigma) = D_\alpha(p,q)$, where
\begin{equation}\label{eq:vectorQRE}
    D_\alpha(p,q) = \frac1{\alpha-1} \log \sum p_i^\alpha q_i^{1-\alpha}.
\end{equation}

\begin{definition}
    Let $\DD_\alpha$ be a family of R\'enyi divergences that fulfills data-processing for a subset $\Lambda\subseteq [0,\infty)$.
    We say that the family $\DD_\alpha$ is {\bf sufficient} if the following holds:
    Let $(\rho_1,\sigma_1)$ and $(\rho_2,\sigma_2)$ be dichotomies on systems $S_1$ and $S_2$ such that $\DD_\alpha(\rho_i, \sigma_i)$ is finite for all $\alpha$ in an open subset interval $(a,b)\subset\Lambda$.
    Then, the dichotomies are interconvertible if and only if all their R\'enyi divergences are equal:
    \[
        (\rho_1,\sigma_1)\leftrightarrow(\rho_2,\sigma_2) \iff \DD_\alpha(\rho_1,\sigma_1)= \DD_\alpha(\rho_2,\sigma_2)\ \ \forall \alpha\in(a,b)\subset\Lambda.
    \]
\end{definition}
In the following, we explore the notion of sufficiency of R\'enyi divergences in the classical and in the quantum case.

\subsection{The classical case}\label{sec:classical}

For classical systems, there is a unique family of R\'enyi divergences given by
\begin{equation}\label{eq:cl_R\'enyi}
    D_\alpha(p,q) = \frac1{\alpha-1}\log \int_X \Big(\frac{dp}{dq}\Big)^\alpha dq\quad\text{if $p\ll q$}.
\end{equation}
where $p$, $q$ are both probability measures on a space $X$. As already stated, the conjecture is true for dichotomies on (possibly continuous) classical systems, which we will prove in this section.
Recall that by $L^1(X,\mu)$, we denote the $\mu$-absolutely continuous measures as well as their $\mu$-densities at the same time.

We briefly comment on the existence of the R\'enyi divergences for continuous distributions.
As the integrand in \cref{eq:cl_R\'enyi} is positive, the R\'enyi divergences $D_\alpha(p,q)$ make sense in $[0,\oo]$ for all $\alpha\in [0,\oo)$ if $p\ll q$.
In the limit $\alpha\to 1$, the R\'enyi divergence agrees with the relative entropy (or Kullback-Leibler divergence) $\int \log(\tfrac{dp}{dq})dp$.
In the following, we use the notation
\begin{equation}
    Q_\alpha(p,q) = e^{(\alpha-1)D_\alpha(p,q)} = \int \Big(\frac{dp}{dq}\Big)^\alpha\, dq.
\end{equation}
These quantities also make sense for all $\alpha<0$ (but may be infinite) if we define $(\tfrac{dp}{dq})^{\alpha}$ as the \emph{pseudo-inverse}.
If one has $p\ll q$ and $q\ll p$ then it follows that $(\tfrac{dp}{dq})^{-\alpha} =(\tfrac{dq}{dp})^\alpha$ for all $\alpha\in \RR$.
In this case it holds that $Q_{\alpha}(p,q)=Q_{1-\alpha}(q,p)$ and $D_\alpha(p,q) = \tfrac\alpha{1-\alpha} D_{1-\alpha}(q,p)$.

\begin{theorem}\label{thm:tilde}
    Let $(X,\mu)$ be a probability space and let $p\ll q$ be probability distributions in $L^1(X,\mu)$.
    Then there is a Borel probability measure $\tilde q$ on $(-\oo,\oo]$ such that 
    \begin{equation}\label{eq:renyi_with_tilde}
        D_\alpha(p,q) = \frac1{\alpha-1}\log \int_{(-\oo,\oo]} e^{-t\alpha}d\tilde q(t)\quad \forall \alpha\in\RR,
    \end{equation}
    where both sides of the equality may be infinite.
    It follows that $\tilde p := e^{-t}\tilde q\ll \tilde q$ is a probability measure and that $(\tilde p,\tilde q)\leftrightarrow(p,q)$ where the second dichotomy is viewed in $L^1((-\oo,\oo],\tilde q)$.
    If the R\'enyi divergences $D_\alpha(p,q)$ exist for an open set of $\alpha\in\RR$, then $\tilde q$ and, hence, $\tilde p$ are uniquely determined by \eqref{eq:renyi_with_tilde}.
\end{theorem}

With this we obtain that the R\'enyi divergences are sufficient for classical systems:

\begin{theorem}[Sufficiency of R\'enyi divergences]\label{thm:classical_solution}
    Let $(X_i,\mu_i)$, $i=1,2$, be $\sigma$-finite measure spaces and $p_i\ll q_i$ probability measures in $L^1(X_i,\mu_i)$.
    If there is an open interval $(a,b)\subset\RR$ such that
    \begin{equation}
        D_\alpha(p_1,q_1) = D_\alpha(p_2,q_2) <\oo \quad\forall\alpha\in(a,b),
    \end{equation}
    there are stochastic maps between $L^1(X_1,\mu_1)$ and $L^1(X_2,\mu_2)$ interconverting the dichotomies $(p_1,q_1)$ and $(p_2,q_2)$.
\end{theorem}

\begin{proof}
    From \cref{thm:tilde} we get two dichotomies $(\tilde p_i,\tilde q_i)$ on $(-\oo,\oo]$.
    From the assumption that the R\'enyi divergences exist (i.e.\ are finite) on an open set of $\alpha\in\RR$ and the uniqueness statement in \cref{thm:tilde} we obtain $(\tilde p_1,\tilde q_1)=(\tilde p_2,\tilde q_2)$ and hence $(p_1,q_1)\leftrightarrow(\tilde p_1,\tilde q_1)=(\tilde p_2,\tilde q_2)\leftrightarrow(p_2,q_2)$.
\end{proof}

For finite-dimensional classical systems, this can be reformulated as:

\begin{corollary}\label{thm:classical_solution_vecs}
    Let $\mathbf p_i,\mathbf q_i\in\RR^{n_i}$, $i=1,2$, be probability vectors (with $\mathbf p_i\ll \mathbf q_i$).
    Then there are stochastic matrices $T\in \RR^{n_1\times n_2}$ and $R\in \RR^{n_2\times n_1}$ so that $T\mathbf p_1 =\mathbf p_2$, $T\mathbf q_1=\mathbf q_2$, $R\mathbf p_2=\mathbf p_1$ and $R\mathbf q_2=\mathbf q_1$ if and only if $D_\alpha(\mathbf p_1,\mathbf q_1)=D_\alpha(\mathbf p_2,\mathbf q_2)$ for an open set of $\alpha$.
\end{corollary}
In \cref{sec:explicit}, we discuss the setting of probability vectors in detail and show how to construct the matrices $T$ and $R$ explicitly.

\begin{lemma}\label{thm:WLOG_mu=q}
    A dichotomy $(p,q)$ in $L^1(X,\mu)$ is interconvertible with the same dichotomy viewed as an element of $L^1(X,q)$ and both have the same R\'enyi divergences.
\end{lemma}

\begin{proof}
    To see the claim we consider the stochastic map $T:L^1(X,q) \to L^1(X,\mu)$ defined by multiplication with the Radon-Nikodym derivative $\tfrac{dq}{d\mu}$.
    This map takes $\tfrac{dp}{dq}\equiv p\in L^1(X,q)$ to $\tfrac{dp}{d\mu}\equiv p\in L^1(X,\mu)$ and $1 \equiv q$ to $\tfrac{dq}{d\mu}\equiv q$.
    A stochastic map in the other direction is given by multiplication with the pseudo-inverse $(\tfrac{dq}{d\mu})^{-1}$.
    This shows that the dichotomies $(p,q)$ viewed as elements of these two different $L^1$ spaces are interconvertible.
    It is also clear that the R\'enyi divergence is the same for both $L^1$-spaces.
\end{proof}

\begin{lemma}\label{thm:injectivity}
    Let $\mu_1,\mu_2$ be finite Borel measures on $\RR$. 
    If there are $a<b \in \RR$ so that
    \begin{equation}
        \int_{-\oo}^\oo e^{-\alpha t}\,d\mu_1(t) = \int_{-\oo}^\oo e^{-\alpha t}\,d\mu_2(t) <\oo \quad \forall\alpha\in (a,b),
    \end{equation}
    then $\mu_1=\mu_2$.
\end{lemma}
\begin{proof}
    Set $\mu=\mu_1-\mu_2$.
    For $\alpha\in(a,b)$ the function $t\mapsto e^{-\alpha t}$ is in $L^1(\RR,|\mu|)$ (because $|\mu|\le\mu_1+\mu_2$).
    The same holds for complex $\alpha$ if $a<\Re(\alpha)<b$.
    The Laplace transform $\int_{-\oo}^\oo e^{-\alpha t} d\mu(t)$ is analytic in the strip with real part in $(a,b)$ and, by assumption, equal to zero on this strip.
    Pick $c\in(a,b)$ and consider the finite Borel measure $\nu=e^{-ct}d\mu(t)$ whose Fourier transform is $\hat\nu(\xi)=\int_{-\oo}^\oo e^{-(c+i\xi)t}d\mu(t)$ and hence vanishes for all $\xi$.
    Injectivity of the Fourier transform implies $\nu=0$ and hence $\mu_1=\mu_2$.
\end{proof}
This Lemma is an injectivity statement for the two-sided Laplace transform. 
While this is surely known, we could not find this precise statement in the literature.

For the proof of \cref{thm:tilde}, we need the concept of the \emph{push-forward measure} from measure theory \cite{tao2011introduction}.
Let $X_1$ and $X_2$ be measurable spaces, and let $f:X_1\to X_2$ be a measurable function.
If $q$ is a measure of $X_1$, then the push-forward measure $f_*(q)$ is defined by $f_*(q)(A) = q(f^{-1}(A))$.
One has $g\in L^1(X_2,f_*(q))$ if and only if $g\circ f \in L^1(X_1,q)$. In this case it holds that
\begin{equation}\label{eq:push-forward-measure}
    \int_{X_2} g \,df_*(q) = \int_{X_1} g\circ f \,dq.
\end{equation}

\begin{proof}[Proof of \cref{thm:tilde}]
    By \cref{thm:WLOG_mu=q} we may assume $(X,\mu)=(X,q)$.
    Set $f = -\log(\tfrac{dp}{dq}) : X\to (-\oo,\oo]$.
    Define $\tilde q:=f_*(q)$ as the push-forward measure of $q$ under $f$.
    Then \eqref{eq:renyi_with_tilde} follows from \eqref{eq:push-forward-measure}:
    \begin{equation}
        Q_\alpha(p,q)=\int_X \Big(\tfrac{dp}{dq}\Big)^\alpha dq
        =\int_X e^{-\alpha f(dp/dq)} dq
        =\int_{-\oo}^\oo e^{-\alpha t}d\tilde q(t)
    \end{equation}
    
    We now show that $(p,q)$ and $(\tilde p,\tilde q)$ are interconvertible.
    We may consider the push-forward as a stochastic map $f_*: L^1(X,q)\to L^1((-\oo,\oo],\tilde{q})$.
    This definition makes use of the fact that all elements in $L^1(X,q)$ may be regarded as measures on $X$ (which are absolutely continuous with respect to $q$).\footnote{The proper definition in terms of densities is $f_*:L^1(X,q)\ni g \mapsto \tfrac{d f_*(gq)}{d\tilde q}\in L^1((-\oo,\oo],\tilde q)$.}
    Then $\tilde p = e^{-t} \tilde q$ is equal to $f_*(p)$ because
    \begin{equation}\label{eq:Radon-Niko}
        f_*(p)(A) = p(f^{-1}(A))=\int_{f^{-1}(A)}\tfrac{dp}{dq} dq = \int_{f^{-1}(A)}e^{-f} dq = \int_A e^{-t} d\tilde q(t)=\tilde p(A)
    \end{equation}
    for all Borel measurable $A\subset \RR$.
    We now define a stochastic map $R : L^1((-\oo,\oo],\tilde q) \to L^1(X,q)$ as $g \mapsto R(g) = g\circ f$.
    This definition uses that elements of $L^1$ spaces are essentially functions (in contrast to the other definitions where we viewed them as measures).
    From the definition, it is clear that $R$ preserves positivity and that $R(g)$ is always measurable.
    Since $\tilde q=f_*(q)$, $R$ satisfies $\int R(g)\,dq=\int g\, d\tilde q$ and hence maps probability densities to probability densities (therefore, $R$ is indeed a stochastic map).
    It remains to be shown that $R$ maps the dichotomy $(\tilde p,\tilde q)$ to $(p,q)$.
    In terms of densities this means that $R(t\mapsto\!e^{-t}) = \tfrac{dp}{dq}$ and $R1 =1$ $q$-almost everywhere. The latter is clear and we check the former:
    \begin{equation}
        \int_A R(t\mapsto\!e^{-t}) dq = \int_A e^{-f(x)}dq(x) = \int_{A} \frac{dp}{dq} dq=p(A).
    \end{equation}
\end{proof}

\subsection{Quantum case: Insufficiency of known quantum R\'enyi divergences}\label{sec:insufficiency_of_known}
We now turn to the quantum case. Let us start by providing examples of quantum R\'enyi divergences. The most well-known families are: 
\begin{enumerate}
\item The \emph{Petz quantum R\'enyi divergence}, defined as
\begin{equation}\label{eq:PetzQRE}
    \petz{\alpha}(\rho, \sigma) = \frac 1 {\alpha-1} \log\tr[\rho^\alpha\sigma^{1-\alpha}],
\end{equation}
where $\sigma^{1-\alpha}$ denotes the \emph{pseudo-inverse} of $\sigma$. $\petz{\alpha}$ satisfies the DPI on $\Lambda=[0,2]$ \cite[Sec.~4.4]{tomamichel}, see also \cref{sec:fdiv}.

\item The \emph{minimal quantum R\'enyi divergence} (or sandwiched R\'enyi divergence) given by
\begin{equation}\label{eq:minQRE}
    D^{\min}_\alpha(\rho,\sigma) = \frac1{\alpha-1}\log \tr\big[ \big(\sigma^{\frac{1-\alpha}{2\alpha} }\rho \sigma^{\frac{1-\alpha}{2\alpha}}\big)^\alpha\big] =: \frac{1}{\alpha-1}\log Q^{\min}_\alpha(\rho,\sigma),
\end{equation}
which satisfies the DPI on $\Lambda = [\tfrac12,\oo)$ \cite[Sec.~4.3]{tomamichel}\cite{berta_variational_2017}. 
The name is referring to the fact that among all families of additive  R\'enyi divergences, the minimal quantum R\'enyi divergence is the smallest \cite[Sec.~4.2]{tomamichel}.

\item The maximal quantum R\'enyi divergence
\begin{equation}\label{eq:maxQRE}
    D_\alpha^{\max}(\rho,\sigma) = \frac1{\alpha-1} \log \tr\big[ \sigma \big(\sigma^{-\frac12}\rho\sigma^{-\frac12} \big)^\alpha \big].
\end{equation}
It satisfies the DPI on $\Lambda=[0,2]$ (and is also only maximal on this region) \cite[Thm.~4.4]{hiai}, \cite[Sec.~4.2.3]{tomamichel}, see also \cref{sec:fdiv}.
\end{enumerate}
All three divergences are additive. Furthermore $D_\alpha^{\min}$ and $\petz{\alpha}$ coincide in the limits $\alpha\rightarrow 1,\infty$ (if the limits exist) \cite{tomamichel}:
\begin{align}
    \lim_{\alpha\rightarrow \infty } D^{\min}_\alpha(\rho,\sigma) =  \lim_{\alpha\rightarrow \infty } \petz{\alpha}(\rho,\sigma) 
    =: D_\infty(\rho,\sigma)
\end{align}
and 
\begin{align}
    \lim_{\alpha\rightarrow 1 } D^{\min}_\alpha(\rho,\sigma) =  \lim_{\alpha\rightarrow 1 } \petz{\alpha}(\rho,\sigma) =  D(\rho,\sigma).
\end{align}
The \emph{quantum max-divergence} $D_\infty(\rho,\sigma)$ can be defined as $D_\infty(\rho,\sigma) = \inf \{\lambda\ :\ \rho \leq \exp(\lambda) \sigma \}$.
In contrast to the minimal and Petz quantum R\'enyi divergences, the maximal quantum R\'enyi divergence has the limit \cite{tomamichel}
\begin{align}
    \lim_{\alpha \rightarrow 1} D^{\max}_\alpha(\rho,\sigma) = \tr[\rho \log\left(\rho^{\frac{1}{2}} \sigma^{-1} \rho^{\frac{1}{2}} \right)] =: D^{\mathrm{BS}}(\rho,\sigma),
\end{align}
where $D^{\mathrm{BS}}(\rho,\sigma)$ is known as Belavkin-Staszewski relative entropy \cite{belavkin_cast_1982}. It fulfills $D(\rho,\sigma) \leq D^{\mathrm{BS}}(\rho,\sigma)$ with equality if and only if the density matrices commute \cite{hiai}.

Apart from the given examples, a variety of quantum R\'enyi divergences has been defined in the literature. 
In \cref{sec:zoo}, we provide an overview of those quantum R\'enyi divergences known to us and summarize their most important properties.

Since quantum R\'enyi divergences fulfill the data-processing inequality, they must be invariant under unitary transformations. 
Importantly, the known R\'enyi divergences (that we discuss in \cref{sec:zoo}) are also invariant under conjugation by an anti-unitary operator:
\begin{align}
\DD_\alpha(\rho,\sigma) = \DD_\alpha(\overline{\rho},\overline{\sigma}) = \DD_\alpha(\rho^\top,\sigma^\top),
\end{align}
where $\overline{A}$ and ${A}^\top$ denote complex conjugation and transposition of an operator $A$, respectively, with respect to some chosen basis. The reason for this is that all of the given divergences are constructed as (optimizations over) traces of matrix functions of $\rho$ and $\sigma$. This property is, therefore, closely related to the fact that one can actually compute the R\'enyi divergences. 
\begin{lemma}\label{lemma:anti-unitary}
    Let $\DD_\alpha$ be a quantum R\'enyi divergence that is anti-unitarily invariant. Then $\DD_\alpha$ is not sufficient.
    \begin{proof}
        Consider a dichotomy $(\rho,\sigma)$ that is \emph{irreducible}, meaning that any operator $A$ for which $[\rho,A]=[\sigma,A]=0$ must be proportional to the identity. (Equivalently, consider a single block in the Koashi-Imoto minimal form.) By anti-unitary invariance we have $\DD_\alpha(\rho,\sigma) = \DD_\alpha(\overline\rho,\sigma)$, where complex conjugation is with respect to the eigenbasis of $\sigma$. If $\DD_\alpha$ was sufficient, \cref{thm:interconvertibility} would imply that there exists a unitary $U$ such that
        \begin{align}
            U\rho U^* = \overline{\rho},\quad U\sigma U^*  = \sigma. 
        \end{align}
        However, in general, such unitaries do not exist. As an example, consider
        \begin{align}
            \rho = \frac{1}{6}\begin{pmatrix} 
            2&i &i\\
            -i& 2& i\\
           -i&-i&2
            \end{pmatrix}
        \end{align}
        in the eigenbasis of $\sigma$ and assume all eigenspaces of $\sigma$ are $1$-dimensional. Then $U = \mathrm{diag}(e^{\mathrm i \phi_1},e^{\mathrm i \phi_2},e^{\mathrm i \phi_3})$, but $U\rho U^* = \overline{\rho}$
        requires
        \begin{align}
            e^{\mathrm i (\phi_1-\phi_2)} = -1,  e^{\mathrm i (\phi_1-\phi_3)} = -1, e^{\mathrm i (\phi_2-\phi_3)} = -1,
        \end{align}
        which is a contradiction since the first two equations imply $e^{\mathrm i (\phi_2-\phi_3)} = +1$.
    \end{proof}
\end{lemma}
\begin{corollary}
    None of the known families of quantum R\'enyi divergence presented in \cref{sec:zoo} is sufficient. Moreover, the whole collection of quantum R\'enyi divergences taken together is also not sufficient. 
\end{corollary}

\subsection{Sufficiency and positive maps: A conjecture}\label{sec:conjecture}
As we saw above, there is a simple way to show that a given quantum R\'enyi divergence is not sufficient: Just check if it is invariant under complex conjugation. This test rules out all of the commonly employed quantum R\'enyi divergences.
Complex conjugation is a particular example of positive, trace-preserving maps and any such map maps quantum states to quantum states. This suggests broadening the perspective and discussing interconvertibility with respect to positive maps instead of completely positive maps. 
In the following, we therefore write
\begin{align}
    (\rho_1,\sigma_1) \ \longleftrightarrow_P (\rho_2,\sigma_2)
\end{align}
if the two dichotomies can be interconverted using positive trace-preserving maps $T,R$.
This generalization, in fact, matches well with the reason why we use completely positive maps to describe quantum processes in the first place: They are required due to the possible presence of entanglement. However, a dichotomy $(\rho,\sigma)$ does not contain any information about how the given system may be correlated to another systems.

Interestingly, several of the known quantum R\'enyi divergences are known to be monotonic under positive trace-preserving maps, at least for certain ranges of their parameters. This includes the Petz R\'enyi divergences, the maximal quantum R\'enyi divergences as well as the minimal quantum R\'enyi divergences. We now explore sufficiency of R\'enyi divergences with respect to interconversion by positive maps.

Let us first observe that the Koashi-Imoto minimal form is not the right minimal form for positive instead of completely positive maps. Consider again the dichotomy $(\rho,\sigma)$ from the proof of \cref{lemma:anti-unitary}. Then we find
\begin{align}
    \left(\frac{1}{2}\rho\oplus \frac{1}{2}\overline{\rho},\frac{1}{2}\sigma\oplus\frac{1}{2}\sigma\right)\ \longleftrightarrow_P\ \left(\frac{1}{2}\rho\oplus \frac{1}{2}\rho,\frac{1}{2}\sigma\oplus\frac{1}{2}\sigma\right) \simeq (\rho\otimes \1/2,\sigma\otimes \1/2)\ \longleftrightarrow (\rho,\sigma).\nonumber
\end{align}
This provides an example where positive maps allow to interconvert a minimal form with two blocks to a minimal form with one block. It is an open problem to find a suitable minimal form for positive maps.\footnote{We suspect that the right minimal form arises from the Koashi-Imoto minimal form by mapping all unitary or anti-unitary equivalent blocks to a single block, but leave it to future work to settle this question.}

In the following, we say that a quantum R\'enyi divergence $\DD_\alpha$ can be \emph{classically simulated} if for 
every dichotomy $(\rho,\sigma)$ there exists classical dichotomy $(p,q)$ such that
\begin{align}
    \DD_\alpha(\rho,\sigma) = D_\alpha(p,q)\quad\forall \alpha\in\Lambda.
\end{align}
\begin{lemma}[\cite{hiai}, see also \cref{sec:fdiv}]
    Both the Petz and the maximal quantum R\'enyi divergences can be classically simulated.
\end{lemma}
For the case of the Petz R\'enyi divergences, the associated classical dichotomies are known as \emph{Nussbaum-Szko\l a distributions} \cite{NussbaumSzkola}.
Recall from \cref{thm:commnogo} and \cref{cor:positive-commnogo} that non-commuting dichotomies cannot be interconverted with classical dichotomies using  positive maps. We thus find:
\begin{corollary}\label{thm:insufficient}
    Neither the Petz quantum R\'enyi divergences nor the maximal R\'enyi divergence are sufficient.
\end{corollary}
In fact, both families arise in the context of \emph{quantum $f$-divergences} \cite{hiai}, all of which can be classically simulated and hence are not sufficient, see \cref{sec:fdiv}. 
To our knowledge, no other quantum R\'enyi divergence is known to be classically simulable.

In particular, for $\alpha\in [1/2,1)\cup(1,\infty)$ we have $D^{\min}_\alpha(\rho,\sigma) \leq \petz\alpha(\rho,\sigma)$ with equality if and only if $[\rho,\sigma]=0$ (as seen above, the two quantities coincide for arbitrary dichotomies in the limits $\alpha\rightarrow 1,\infty$). 
This follows from the Araki-Lieb-Thirring inequality \cite{araki_inequality_1990,lieb_inequalities_2005} and its equality conditions \cite{hiai_equality_1994}, see also \cite{datta_limit_2014,sutter_approximate_2018}.
Since we cannot use the Nussbaum-Szko\l a distributions to show an analogous statement for $D^{\min}_\alpha$, this family might, therefore, still be sufficient.
We will see below that $D^{\min}_\alpha$ is not classically simulable by constructing a restricted class of non-commuting states for which there cannot exist any (sufficiently regular) probability measure with matching R\'enyi divergences. This leads us to our conjecture:
 
\begin{conjecture}\label{con:explicit_form}
    Let $(\rho_1,\sigma_1)$ and $(\rho_2,\sigma_2)$ be pairs of density operators on quantum system $S_1$ and $S_2$.
    Let $(a,b)$, $\tfrac12\leq a<b$, be any interval on which the minimal quantum R\'enyi divergences of both dichotomies are finite.
    Then the dichotomies are interconvertible via positive, trace-preserving maps if and only if they have the same minimal quantum R\'enyi divergences on this interval, i.e.,
    \begin{equation}\label{eq:explicit_form}
        (\rho_1,\sigma_1)\leftrightarrow(\rho_2,\sigma_2) \iff D^{\min}_\alpha(\rho_1,\sigma_1)=D^{\min}_\alpha(\rho_2,\sigma_2)<\infty \ \ \forall \alpha\in(a,b).
    \end{equation}
\end{conjecture}
As in the classical case, a positive resolution of the conjecture directly implies a corresponding no-go result for catalysis, because $D^{\min}_\alpha$ is additive over tensor products:
\begin{corollary}[No catalysis]\label{cor:no-cat}
    If \cref{con:explicit_form} is true and $D^{\min}_\alpha(\rho_1\otimes \omega,\sigma_1\otimes \chi) <\infty$ for $\alpha\in(a,b)$,
    \begin{align}
        (\rho_1\otimes \omega,\sigma_1\otimes \chi)\ \longleftrightarrow_P\ (\rho_2\otimes\omega,\sigma_2\otimes\chi) \quad \Rightarrow\quad (\rho_1,\sigma_1) \ \longleftrightarrow_P\ (\rho_2,\sigma_2).
    \end{align}
\end{corollary}

The no catalysis corollary offers a way to falsify our conjecture: It is sufficient to find a single non-trivial example of catalysis for interconvertibility via positive maps to rule out sufficiency for any additive divergence. 
On the other hand, we believe that it may be possible to prove a no-go result for catalysis in the quantum case without first proving \cref{con:explicit_form}, possibly by making use of the normal form for interconvertible dichotomies, see \cref{sec:structure}.
Such a result would certainly constitute evidence in favor of our conjecture.
However, we have not been able to prove such a result.
In fact, even in the classical case, we are not aware of a proof that does not already also essentially imply the sufficiency of R\'enyi divergences. We therefore leave it as an open problem to prove that there is no catalysis for the interconvertibility of quantum dichotomies via positive maps.

We also leave as an open problem to find a way to create numerical examples of dichotomies $(\rho_1,\sigma_1)$ and $(\rho_2,\sigma_2)$ for which $D^{\min}_\alpha(\rho_1,\sigma_1) = D^{\min}_\alpha(\rho_2,\sigma_2)$ for all $\alpha\in\Lambda$, but for which one does not know a priori that $(\rho_1,\sigma_1)\leftrightarrow_P (\rho_2,\sigma_2)$.

In general $D^{\min}_\alpha(\rho,\sigma) \neq D^{\min}_\alpha(\sigma,\rho)$. We have stated the conjecture with only one ordering of the states in the R\'enyi divergences, as this is sufficient in the classical case. Indeed, for commuting states, we have 
\begin{align}
D^{\min}_\alpha(\rho,\sigma) = \frac{\alpha}{1-\alpha}D^{\min}_{1-\alpha}(\sigma,\rho)
\end{align}
if  $\rho\ll\sigma\ll \rho$.  Therefore, as long as the R\'enyi divergence is finite for $\alpha\in(0,1)$, it is not necessary to consider both orderings explicitly. 
It can be shown that this relation holds \emph{if and only if} $[\rho,\sigma]=0$ using the Araki-Lieb-Thirring inequality and its equality condition. 
In particular, if one demands $D^{\min}_\alpha(\rho_1,\sigma_1) = D^{\min}_\alpha(\rho_2,\sigma_2)$ \emph{and} $D^{\min}_\alpha(\sigma_1,\rho_1) = D^{\min}_\alpha(\sigma_2,\rho_2)$, then $[\rho_1,\sigma_1]=0$ implies $[\rho_2,\sigma_2]=0$. We prove a (much) weaker statement of this form, see \cref{thm:R\'enyis_commutativity} below.
This may indicate that the corresponding quantum result requires a condition about equality of R\'enyi divergences with respect to both orderings.

We will now present affirmative results to weaker statements supporting the idea that \cref{con:explicit_form} should be true.
We say that the conjecture holds for two pairs of states (or for a certain class of pairs of states) if \eqref{eq:explicit_form} holds for these states. 
As a corollary of the complete solution for the classical case, we get:

\begin{lemma}\label{thm:commuting_q_states}
    The conjecture holds for commuting pairs of states.
\end{lemma}

Under the additional assumption that one state is pure, the conjecture is true:

\begin{lemma}\label{thm:solution_pure_case}
    The conjecture is true for pairs $(\rho_1,\sigma_1)$ and $(\rho_2,\sigma_2)$ of density operators if the $\rho_i = \kettbra{\psi^{(i)}}$ are pure states.
\end{lemma}

\begin{lemma}\label{thm:R\'enyis_commutativity}
    Let $(\rho_1,\sigma_1)$ and $(\rho_2,\sigma_2)$ be pairs of density operators with equal minimal quantum R\'enyi divergences on finite-dimensional quantum systems.
    If $(\rho_1,\sigma_1)$ commute and $\rho_2$ is pure, the conjecture is true.
\end{lemma}

\begin{remark}\label{rem:salt}
    We comment on the implications of \cref{thm:solution_pure_case,thm:R\'enyis_commutativity}:
    One should take \cref{thm:solution_pure_case} with a grain of salt as it is also true for the Petz quantum R\'enyi divergences (with a similar proof).
    \cref{thm:R\'enyis_commutativity} however is false for the Petz quantum R\'enyi divergence (and by the same argument also for $\geom\alpha$).
    This is a consequence of the existence of the Nussbaum-Sko\l a distributions $P$ and $Q$ as defined in \cref{sec:fdiv}.
    As noted before (see \cref{thm:insufficient} and after), their Petz quantum R\'enyi divergences coincide with those of $(\rho,\sigma)$ but they cannot be interconvertible.
    Since for non-commuting $\rho$ and $\sigma$ the distribution $P$ is never pure and $\petz\alpha(\rho, \sigma) > \sand\alpha(\rho,\sigma)$ this is not in conflict with the validity of \cref{thm:solution_pure_case} for $\petz\alpha$ nor with \cref{thm:R\'enyis_commutativity}.
\end{remark}

\begin{remark}[$\alpha$-$z$ R\'enyi divergences]
    The $\alpha$-$z$ quantum R\'enyi divergence generalizes the minimal and Petz quantum R\'enyi divergences, see \cref{sec:zoo}. If we formulate the conjecture in terms of the $\alpha$-$z$ R\'enyi divergence instead of $D^{\min}_\alpha$ we obtain a stronger statement than \cref{thm:R\'enyis_commutativity}: Suppose that $D_{\alpha,z}(\rho_1,\sigma_1)=D_{\alpha,z}(\rho_2,\sigma_2)$ for the values $(\alpha,z) = (\alpha^*,\alpha^*)$ and $(\alpha,z)=(\alpha^*,1)$  with $\alpha^* \in [1/2,1)$ and moreover $[\rho_1,\sigma_1]=0$. Then we can deduce already that $[\rho_2,\sigma_2]=0$ even if $\rho_2$ is not pure. To see this, note that (since $[\rho_1,\sigma_1]=0$)
    \begin{align}
        D_{\alpha,z}(\rho_1,\sigma_1) = \petz{\alpha}(\rho_1,\sigma_1) = D^{\min}_{\alpha}(\rho_1,\sigma_1)
    \end{align}
    and $D_{\alpha,\alpha}(\rho_2,\sigma_2) = D^{\min}_\alpha(\rho_2,\sigma_2), D_{\alpha,1}(\rho_2,\sigma_2)=\petz{\alpha}(\rho_2,\sigma_2)$. 
    Therefore
    \begin{align}
        \petz{\alpha^*}(\rho_2,\sigma_2) = \petz{\alpha^*}(\rho_1,\sigma_1) = D^{\min}_{\alpha^*}(\rho_1,\sigma_1) = D^{\min}_{\alpha^*}(\rho_2,\sigma_2),
    \end{align}
    which is only possible if $[\rho_2,\sigma_2]=0$. 
    It is therefore conceivable that \cref{con:explicit_form} is true for the $\alpha$-$z$ quantum R\'enyi divergence but not for the minimal quantum R\'enyi divergence.
\end{remark}

\cref{thm:solution_pure_case,thm:R\'enyis_commutativity} both deal with a dichotomy $(\rho,\sigma)$ where $\rho$ is pure.
Let us, therefore, first evaluate the minimal quantum R\'enyi divergence for this special case.
It is useful to define the quantity
\begin{equation}
    Q^{\min}_\alpha(\rho,\sigma) = \tr\big[\big(\sigma^{\frac{1-\alpha}{2\alpha}}\rho\sigma^{\frac{1-\alpha}{2\alpha}}\big)^\alpha\big] = \tr\big[\big(\rho^{\frac12} \sigma^{\frac{1-\alpha}\alpha}\rho^{\frac12}\big)^\alpha\big],
\end{equation}
so that $D_\alpha^{\min}(\rho,\sigma) = (\alpha-1)^{-1}\log Q^{\min}_\alpha(\rho,\sigma)$.
The second equality is seen as follows: Set $A=\sigma^{\frac{1-\alpha}{2\alpha}}\rho^{\frac12}$ and use that the singular value decomposition implies $\tr[(A^*A)^\alpha] = \tr[(AA^*)^\alpha]$.
If $\rho=\kettbra\psi$ is pure and $\sigma = \sum v_i P_i$ is its spectral resolution, then \begin{align}\label{eq:QalphaPure}
    Q_\alpha^{\min}(\rho,\sigma) = \ip\psi{ \sigma^{\frac{1-\alpha}\alpha}\psi}^\alpha = \big(\sum u_i v_i^{\frac{1-\alpha}\alpha}\big)^\alpha = \big(\sum \frac{u_i}{v_i} v_i^{1/\alpha}\big)^\alpha,
\end{align}
where we used the shorthand $p_i=\ip\psi{P_i\psi}$.

\begin{proof}[Proof of Lemma~\ref{thm:solution_pure_case}] For $j=1,2$, let $\rho_j=\proj{\psi^{(j)}}$ with $\ket{\psi^{(j)}}\in \mc H_j$ and let
\begin{equation}\label{eq:sigma_q}
    \sigma_j=\sum_{i=1}^{n_j} v_i^{(j)} P_i^{(j)}
\end{equation}
with $P_i^{(j)}$ being the spectral projectors of $\sigma_j$ and $\{v_i^{(j)}\}_{i=1}^{n_j}$ the spectrum of $\sigma_j$.
Note that $n_j$ counts the number of distinct eigenvalues of the state $\sigma_j$ and not the dimension of $\mc H_j$. 
We further define the sub-normalized vectors $\ket{\psi_{i}^{(j)}} = P_i^{(j)}\ket{\psi^{(j)}}$ with associated norms
\begin{align}\label{eq:normsPure}
    u^{(j)}_i = \big\|\ket{\psi_i^{(j)}}\big\|^2 = \bra{\psi^{(j)}} P_i^{(j)}\ket{\psi^{(j)}}.
\end{align}
Finally, denote by $R^{(j)}$ the othonormal projectors onto $\mathrm{span}\{\ket{\psi_i^{(j)}}\}_{i=1}^{n_j}$. 
By construction we have $[R^{(j)},\sigma^{(j)}]=0$ and $R^{(j)}\ket{\psi^{(j)}} = \ket{\psi^{(j)}}$.
Now let $D_\alpha^{\min}(\rho_1,\sigma_1) = D_\alpha^{\min}(\rho_2,\sigma_2)$ for all $\alpha$ in some open subset of $(1,\infty)$. Using \eqref{eq:QalphaPure} we obtain
\begin{align}
   \sum_i \left(\frac{u^{(1)}_i}{v^{(1)}_i}\right) {q^{(1)}_i}^{\frac{1}{\alpha}} = \sum_i \left(\frac{u^{(2)}_i}{v^{(2)}_i}\right) {v^{(2)}_i}^{\frac{1}{\alpha}}
\end{align}
on some non-empty open interval of $(1,\infty)$.
Both sides are analytic in $\alpha$ and can be analytically continued to all $\alpha>0$. The equality then extends to all $\alpha>0$.
Substituting $1/\alpha\mapsto \alpha$, we get for all $\alpha>0$
\begin{align}
    \sum_{i=1}^{n_1} \left(\frac{u^{(1)}_i}{v^{(1)}_i}\right) {v^{(1)}_i}^{\alpha}= \sum_{i=1}^{n_2} \left(\frac{u^{(2)}_i}{v^{(2)}_i}\right) {v^{(2)}_i}^{\alpha}.
\end{align}
Since exponential functions $x\mapsto \exp(a x)$ with distinct $a$ are linearly independent and the $v^{(j)}_i$ are distinct for fixed $j$, we find that $n_1=n_2$ and that there exists a permutation $\pi$ such that:
\begin{align}\label{eq:coeffsPure}
     v^{(2)}_j = v^{(1)}_{\pi(j)},\quad u^{(2)}_j = u^{(1)}_{\pi(j)}.
\end{align}
From now on, we relabel the indices so that $\pi=\id$. By construction we have
\begin{align}
    R^{(1)}\sigma_1 = \sum_j \frac{v_j^{(1)}}{u_j^{(1)}}\proj{\psi_j\up 1},\quad R^{(2)}\sigma_2 = \sum_j\frac{v_j^{(2)}}{u_j^{(2)}}\proj{\psi_j\up 2}.
\end{align}
Eqs.~\eqref{eq:normsPure} and \eqref{eq:coeffsPure} imply that there exists a unitary operator $U:R^{(1)}\mc H_1\rightarrow R^{(2)}\mc H_2$ such that
\begin{align}
        U\ket{\psi^{(1)}} = \ket{\psi^{(2)}},\quad U R^{(1)}\sigma_1 R^{(1)}U^* = R^{(2)} \sigma_2 R^{(2)}.
    \end{align}
    Furthermore
    \begin{align}
        \tr[(\1-R^{(1)})\sigma_1] = 1 - \sum_j v_j^{(1)} = 1-\sum_j v_j^{(2)} =\tr[(\1-R^{(2)})\sigma_2].
    \end{align}
    Hence there exist trace-preserving completely positive maps $\tilde T:\T((\1-R^{(1)})\mc H_1)\rightarrow \T((\1-R^{(2)})\mc H_2)$ and $\tilde S:\T((\1-R^{(2)})\mc H_2)\rightarrow \T((\1-R^{(1)})\mc H_1)$ such that
    \begin{align}
        \tilde T((\1-R^{(1)})\sigma_1) &= (\1-R^{(2)})\sigma_2\\
        \tilde S((\1-R^{(2)})\sigma_2) &= (\1-R^{(1)})\sigma_1.
    \end{align}
    We can hence define the quantum channels $T:\T(\mc H_1)\rightarrow \T(\mc H_2)$ and $S:\T(\mc H_2)\rightarrow \T(\mc H_1)$ by
    \begin{align}
        T(\tau) := U R^{(1)} \tau R^{(1)} U^* + \tilde T((\1-R^{(1)})\tau (\1-R^{(1)}))\\
        S(\varphi) := U^* R^{(2)} \varphi R^{(2)} U + \tilde S((\1-R^{(2)})\varphi(\1-R^{(2)})), 
    \end{align}
    which fulfill
    \begin{align}
        T(\rho_1) = \rho_2, \quad S(\rho_2)=\rho_1,\quad T(\sigma_1)=\sigma_2,\quad S(\sigma_2)=\sigma_1.
    \end{align}
\end{proof}
\begin{remark}
    The above proof of \cref{thm:solution_pure_case} can also be carried out if $\mc H_j$ are infinite-dimensional provided that $D_\alpha(\rho_j,\sigma_j)$ are finite for all $\alpha>1$. This is not always the case. 
    To see this, first note that any two discrete probability distributions $u,v\in \ell^1$ with $u\ll v$ may arise from a suitable pure state $\ket\psi$ and mixed state $\sigma$ as in \eqref{eq:QalphaPure}.
    For example, consider
    \begin{align}
        u_i = \frac{6}{\pi^2}\frac{1}{i^2},\quad v_i = \frac{1}{\zeta(2+s)}\frac{1}{i^{2+s}},
    \end{align}
    where $\zeta$ is the Riemann-Zeta function. Then 
    \begin{align}
        D^{\min}_\alpha(\rho,\sigma) = \log\big(\zeta(2+s)\big) + \frac{\alpha}{\alpha-1}\log\big(\frac{6}{\pi^2} \sum_{i=1}^\infty i^{-\frac{2+s-\alpha s}{\alpha}}\big),
    \end{align}
    which diverges for $\alpha \geq \frac{2+s}{1+s}$. Incidentally, the Petz quantum R\'enyi divergence, in this case, is given by
    \begin{align}
        \petz\alpha(\rho,\sigma) = \log\big(\zeta(2+s)\big) + \frac{1}{\alpha-1}\log\big(\frac{6}{\pi^2}\sum_{i=1}^\infty i^{-(2+(1-\alpha)(2+s))}\big),
    \end{align}
    which converges if and only if $\alpha<\frac{3+s}{2+s}$. Thus, for $\alpha\in(\frac{3+s}{2+s},\frac{2+s}{1+s})$ the Petz quantum R\'enyi divergence is infinite while the minimal quantum R\'enyi divergence is finite, showing that their difference is in general unbounded.
\end{remark}

\cref{thm:R\'enyis_commutativity} will follow from the following more general statement:

\begin{lemma}\label{thm:classical_R\'enyis2}
    Let $\rho\ll\sigma$ be finite-dimensional quantum states with $\rho$ being pure and $[\rho,\sigma]\neq 0$.
    Then there is no classical system with probability distributions $p\ll q$ and $(\tfrac{dp}{dq})^{\pm1}\in L^\oo(q)$ (pseudo-inverse!) with $D_\alpha(p,q) = D_\alpha^{\min}(\rho,\sigma)$.
\end{lemma}

For the proof of \cref{thm:classical_R\'enyis2} we use the following argument: 
The function $Q_\alpha^{\min}(\rho,\sigma)$ can be extended to a meromorphic function on $\CC$, which has an essential singularity at the origin if and only if $\rho$ and $\sigma$ do not commute.
The assumption that the Radon-Nikodym derivative is essentially bounded from above and below will, however, imply that $Q_\alpha(p,q)$ is holomorphic in the origin.
This shows that $\rho$ and $\sigma$ commute so that interconvertibility follows from \cref{thm:commuting_q_states}.

\begin{proof}
That $Q_\alpha(\rho,\sigma)$ indeed extends to a meromorphic function is seen from \cref{eq:QalphaPure}.
Also note that $\rho$ and $\sigma$ commute if and only if exactly one of the $u_i$'s is equal to one (so that the others are zero by \eqref{eq:sigma_q}). This is, in turn, equivalent to $Q_\alpha(\rho,\sigma)$ being proportional to an exponential function.
The limits of $Q_\alpha(\rho,\sigma)$ as $\alpha \to 0^\pm$ are
\begin{align}
    \lim_{\alpha\to0^\pm}\Big(\sum \frac{u_i}{v_i} v_i^{1/\alpha} \Big)^\alpha  = v_\pm\cdot\lim_{\alpha\to0^\pm}\Big(\sum \frac{u_i}{v_i} \Big(\frac{v_i}{v_\pm}\Big)^{1/\alpha}\Big)^\alpha = v_\pm,
\end{align}
where $v_+$ is the largest and $v_-$ is the smallest number among the $v_i$ (i.e.\ the maximal, resp.\ minimal, eigenvalue of $\sigma$).
Therefore, $Q_\alpha(\rho,\sigma)$ has an essential singularity at the origin if $v_+ > v_-$ which is equivalent to non-commutativity of $\rho$ and $\sigma$.
In the commuting case $Q_\alpha(\rho,\sigma)$ is an entire function and has no singularity.
Therefore, having the essential singularity at the origin is indeed equivalent to commutativity of $\rho$ and $\sigma$.

Now let $p\ll q$ be probability distributions on a $\sigma$-finite measure space such that $(\tfrac{dp}{dq})^{\pm 1}\in L^\oo(q)$.
We will show that $Q_\alpha(p,q)$ has an analytic extension to all of $\CC$.
This holomorphic extension is defined by allowing complex $\alpha$ in the definition $Q_\alpha(p,q) = \int (\tfrac{dp}{dq})^\alpha dq$.
These integrals exist in $L^1$ for all complex $\alpha = \pm a +ib$, $a>0$, $b\in\RR$, because $\int |(\tfrac{dp}{dq})^\alpha| dq =\int |\tfrac{dp}{dq}|^{\pm a} dq \le \big\|{(\tfrac{dp}{dq})^{\pm1}}\big\|_{L^\oo(q)}^a <\oo$.
\end{proof}

\subsection{Complete monotonicity}
\label{sec:CM}

A non-constant smooth function $f:(0,\oo)\to\RR$ is called \emph{completely monotone}, if 
\begin{equation}
    (-1)^n \frac{d^n}{d\alpha^n}f(\alpha) >0, \quad \forall n\in\NN,\alpha\in\RR.
\end{equation}
Bernstein's theorem \cite{CMfunctions} states that completely monotone functions are precisely the Laplace transforms of positive Borel measures $\mu$ on $[0,\oo)$, i.e.,
\begin{equation}\label{eq:CMfunction}
    f(\alpha) = \int_0^\oo e^{-t\alpha}\,d\mu(t), \quad\alpha>0.
\end{equation}
One has $\lim_{\alpha\to0^+}f(\alpha)=\mu(\{0\})$ and this limit exists in $(0,\oo]$.
Note that multiplying $f$ by an exponential corresponds to a shift of the measure $\mu$: $f(\alpha)e^{s\alpha} = \int_0^\oo e^{-(t-s)\alpha} \,d\mu(t)= \int_0^\oo e^{-t\alpha}\,d\mu(t+s)$.

For probability distributions $p\ll q$ on a $\sigma$-finite measure space $(X,\mu)$ the argument used to prove the classical case constructs a probability measure $\tilde q$ on $(-\oo,\oo]$ so that the two-sided Laplace transform of $\tilde q$ is $Q_\alpha(p,q)$.
The support of this measure $\tilde q$ is contained in $[-D_\oo(p,q),\oo]$.
If now $D_\oo(p,q) = \log \|\tfrac{dp}{dq}\|_{L^\oo(q)}$ is finite, then we can multiply $Q_\alpha(p,q)$ by $e^{\alpha D_\oo(p,q)}$ to obtain a completely monotone function.
With this idea, we obtain:

\begin{proposition}\label{thm:CM}
    Let $\rho\ll\sigma$ be quantum states and let $\DD_\alpha$ be a family of quantum R\'enyi divergences with $\lim_{\alpha\to\oo}\DD_\alpha(\rho,\sigma)=D_\oo(\rho,\sigma)<\oo$.
    Then there are classical probability distributions $p\ll q$ on a $\sigma$-finite measure space such that $\DD_\alpha(\rho,\sigma)=D_\alpha(p,q)$ if and only if the function
    \begin{equation}
        f(\alpha) = e^{(\alpha-1)(\DD_\alpha(\rho,\sigma)-D_\oo(\rho,\sigma))}
    \end{equation}
    is completely monotone.
\end{proposition}
\begin{proof}
    The only if part of the proof was explained in the text before the proposition.
    The assumption implies that the function $g(\alpha) = e^{-D_\oo(\rho,\sigma)}f(\alpha)$ is also completely monotone.
    Let $\mu$ be the measure defined by the function $g$ through \cref{eq:CMfunction}.
    Define a probability measures $q$ and $p$ on $[-D_\oo(\rho,\sigma),\oo]$ by $q(t) = \mu(t+D_\oo(p,q)) + (1-\mu((-\oo,\oo)))\delta_\oo(t)$ and $dp(t) = e^{-t}dq(t)$.
    The latter is normalized since 
    $$
        p\big([-D_\oo(\rho,\sigma),\oo]\big) = (1-\mu((-\oo,\oo)))+ \int_{-D_\oo(\rho,\sigma)}^\oo e^{-t}dq(t) =1.
    $$
    We set $\mathbb Q_\alpha(\rho,\sigma)=e^{(\alpha-1)\DD_\alpha(\rho,\sigma)}$. Then
    \begin{multline*}
        Q_\alpha(p,q) = \int_{-D_\oo(\rho,\sigma)}^\oo e^{-\alpha t}\,dq(t) 
        = e^{\alpha D_\oo(\rho,\sigma)} \int_0^\oo e^{-\alpha t}\,d\mu(t) \\
        = e^{\alpha D_\oo(\rho,\sigma)} g(\alpha) 
        = e^{(1-\alpha)\DD_\alpha(\rho,\sigma)} = \mathbb Q_\alpha(\rho,\sigma).
    \end{multline*}
\end{proof}

We apply \cref{thm:CM} to show that if \cref{con:explicit_form} is true, there are infinitely many inequalities on the minimal quantum R\'enyi divergence which separate the quantum and classical case.
Furthermore, we give an alternative proof of our claim that the maximal quantum R\'enyi divergence is not complete.
We start with the latter:

If $\rho =\kettbra\psi$ is pure and $\sigma$ be arbitrary with $\psi\in\mathrm{supp}(\sigma)$ then 
$$
    e^{(\alpha-1)D^{\max}_\alpha(\rho,\sigma)} 
    = \tr\big[ \sigma \big( \sigma^{-1/2}\rho\sigma^{-1/2}\big)^\alpha \big]
    = \|\sigma^{-1/2}\psi\|^{2(\alpha-1)}.
$$
That this is exponential in $\alpha$ implies that the appropriately translated version of $Q_\alpha$ is completely monotone so that classical distributions with the same R\'enyi divergences exist.
Just as in the other proof, this contradicts interconvertibility \cref{thm:noteleportation} whenever $\rho$ and $\sigma$ do not commute.

We now turn to the inequalities separating commutativity and non-commutativity.
If a family $\DD_\alpha$ of R\'enyi divergences is sufficient, it holds for a dichotomy $(\rho,\sigma)$ with $\rho\ll\sigma$ and $\DD_\oo(\rho,\sigma)<\oo$ that
\begin{equation}\label{eq:completely-monotone-Dalpha}
    [\rho,\sigma]=0
    \iff (-1)^n\frac{d^n}{d\alpha^n}\Big(e^{(\alpha-1)(\DD_\alpha(\alpha)-\DD_\oo(\rho,\sigma))}\Big)\ge0 \quad \forall\alpha>0,\, n\in\NN.
\end{equation}
Notice also that each of the functions $(-1)^n \tfrac{d^n}{d\alpha^n}(\ldots)$ is completely monotone in $\alpha$. Evaluating \eqref{eq:completely-monotone-Dalpha} for $n=1$, we find
\begin{align}
    (\alpha-1) \partial_\alpha \DD_\alpha(\rho,\sigma) < \DD_\oo(\rho,\sigma) - \DD_\alpha(\rho,\sigma).
\end{align}
Many more point-wise inequalities are known for completely monotone functions \cite{fink_kolmogorov-landau_1982,mitrinovic_inequalities_1993,CMfunctions}. 
For example, suppose that $\mathbf n,\mathbf m \in \NN^L$ are two non-increasingly ordered vectors ($n_1\geq n_2\geq \cdots$) such that $\mathbf n \succ \mathbf m$, i.e., $\sum_{j=1}^k n_j \geq \sum_{j=1}^k m_j$ for all $k=1,\ldots,L-1$ and $\sum_{j=1}^L n_j = \sum_{j=1}^L m_j$. 
Then \cite{fink_kolmogorov-landau_1982}
\begin{align}
\prod_{j=1}^L (-1)^{n_j} f^{(n_j)}(x) \geq \prod_{j=1}^L (-1)^{m_j} f^{(m_j)}(x)
\end{align}
for any completely monotone function $f$. In other words, the function 
\begin{align}
    \mathbf n \mapsto \prod_{j=1}^L (-1)^{n_j} f^{(n_j)}(x) 
\end{align}
is Schur-convex. For $L=2$ this yields log-convexity of completely monotone functions.
Exploring the resulting inequalities for R\'enyi divergences seems like an interesting open problem. 

As discussed above, if Conjecture~\ref{con:explicit_form} is true, then the function
\begin{align}\label{eq:g}
g(\alpha|\rho,\sigma) :=   e^{-\alpha D_\oo(\rho,\sigma)} e^{(\alpha-1)D^{\min}_\alpha(\rho,\sigma)} = e^{-\alpha D_\oo(\rho,\sigma)} Q^{\min}_\alpha(\rho,\sigma)
\end{align} 
is completely monotone if and only if $[\rho,\sigma]=0$. This allows us to obtain additional evidence in favor of the conjecture by numerical means. Specifically, we can compute $g(\alpha|\rho,\sigma)$ and its derivatives for a choice of $(\rho,\sigma)$ numerically and see whether \eqref{eq:completely-monotone-Dalpha} holds true. 
We have checked this for numerous examples without finding a counter-example to our conjecture. A numerical example for qubits is presented in Fig.~\ref{fig:qubits}. Note that these numerical results also provide a different means to see that no equivalent of Nussbaum-Szko\l a distributions exist for the minimal quantum R\'enyi divergence.
\begin{figure}
    \centering
    \includegraphics[width=\textwidth]{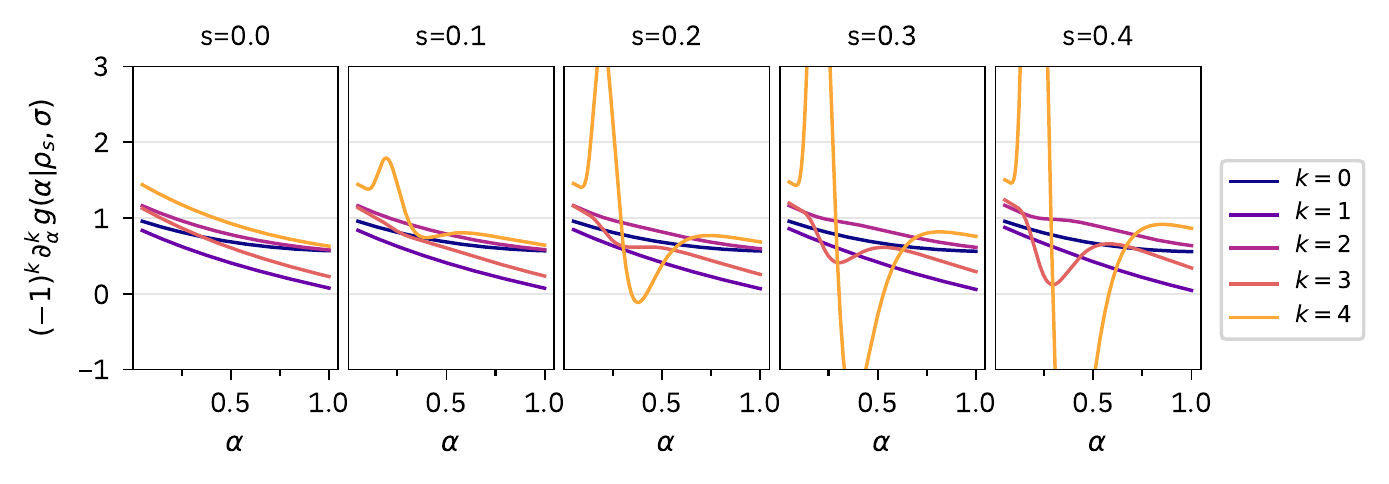}
    \caption{The derivatives $(-1)^k \partial_\alpha^k\,g(\alpha|\rho_s,\sigma)$ for $\DD_\alpha = D^{\min}_\alpha$ in the case $\rho_s = \mathrm{e}^{s \sigma_x}/2 \cosh(s),\sigma = \mathrm{e}^{\sigma_z}/\cosh(1)$ and for various values of $s$ (see Eq.~\eqref{eq:g} for the definition of $g(\alpha|\rho,\sigma)$). If $g(\alpha|\rho_s,\sigma)$ was completely monotone, all curves would have to be monotonically decreasing and convex. Note that $[\rho_s,\sigma]=0$ if and only if $s=0$. For $s=0$ the function is completely monotone. As $s$ increases, the commutator-norm $\norm{[\rho_s,\sigma]}$ increases monotonically while lower and lower derivatives signal that the function is not completely monotone. Already for $s=0.3$ one can see that the second derivative is not completely monotone.}
    \label{fig:qubits}
\end{figure}

\subsection*{Acknowledgements}
We would like to thank Bjarne Bergh, Milan Mosonyi, Robert Salzmann, Joseph Schindler, Alexander Stottmeister, David Sutter, Marco Tomamichel, Reinhard F. Werner and Andreas Winter for interesting and helpful discussions. 
N.G. acknowledges financial support by the MICIIN with funding from European Union NextGenerationEU (PRTR-C17.I1) and by the Generalitat de Catalunya.
H.W. acknowledges support by the DFG through SFB 1227 (DQ-mat), Quantum Valley Lower Saxony, and funding by the Deutsche Forschungsgemeinschaft (DFG, German Research Foundation) under Germanys Excellence Strategy EXC-2123 QuantumFrontiers 390837967. 

\appendix

\section{Interconvertibility from R\'enyi Divergences without absolute continuity}
\label{sec:no-abs-continuity}

In the main text, we only considered classical dichotomies with $p\ll q$. 
In this appendix, we generalize \cref{thm:classical_solution} to arbitrary pairs of probability measures. To do this, we first have to find a way to define R\'enyi divergences for such measures. We appeal to the Lebesgue decomposition theorem stating that if $p$ and $q$ are probability measures on a measure space, then there is a unique decomposition
\begin{align}
p = \lambda p^\| + (1-\lambda) p^\perp,\quad p^\|\ll q,\quad p^\perp \perp q,\quad \lambda \in [0,1],
\end{align}
with $p^\|$ and $p^\perp$ probability measures as well. Now, for $\alpha\in(0,1)\cup(1,\oo)$ define \cite{hiai}
\begin{align}\label{def:better_def}
D_\alpha(p,q) :=\begin{cases} \frac{\alpha}{\alpha-1}\log(\lambda) + D_\alpha(p^\|,q),\ &\text{if}\ \alpha \in (0,1)\\
D_\alpha(p,q),&\text{if}\ p\ll q,\ \alpha>1,\\
\oo &\text{otherwise}.
\end{cases}
\end{align}
This reduces to the usual definition for $p=p^\| \ll q$. 
Moreover, note that for $\alpha \in (0,1)$ we have $D_\alpha(p,q) = 0$ precisely if $\lambda=1$ and hence $p=p^\|$. It thus follows that $p=q$.
Our main result now generalizes as follows:
\begin{theorem}\label{thm:classical_solution_without_ac}
        Let $(X_i,\mu_i)$, $i=1,2$, be measured spaces and $p_i,q_i$ probability measures in $L^1(X_i,\mu_i)$ such that $p_i$ is not orthogonal to $q_i$.
        If there is an open interval $(a,b)\subset \RR$ such that
        \begin{equation}
            D_\alpha(p_1,q_1) = D_\alpha(p_2,q_2) <\oo \quad\forall\alpha\in(a,b),
        \end{equation}
        there are stochastic maps between $L^1(X_1,\mu_1)$ and $L^1(X_2,\mu_2)$ interconverting the dichotomies $(p_1,q_1)$ and $(p_2,q_2)$.
    \end{theorem}
    Some remarks are in order:
    First, if $p_1\perp q_1$ and $p_2\perp q_2$ the two dichotomies are trivially interconvertible. Therefore we don't loose much by excluding this case (one could incorporate it by considering $Q_\alpha$ instead of $D_\alpha$).

    Second, let us comment on the role of the interval $(a,b)$:
    Suppose that $p_i$ is not absolutely continuous with respect to $q_i$. Then the conditions of the Theorem require $0\leq a < b \leq 1$, because otherwise the R\'enyi divergence is infinite by definition. 
    Conversely, if  $p_i\ll q_i$ and $b>1$, then we already know from \cref{thm:classical_solution} that the R\'enyi divergences must match for all $0<\alpha<b$, since the two dichotomies are interconvertible and the R\'enyi-divergences are non-decreasing in $\alpha$.
    Thus we can in fact always assume $(a,b)\subseteq (0,1)$. 

    \begin{proof}[Proof sketch]
    We only need to discuss the case $\lambda<1$. 
    We thus get as input $D_\alpha(\lambda_1 p_{1}^\|,q_1) = D_\alpha (\lambda_2 p_2^\|,q_2)$. 
    One can now check that the proof of \cref{thm:classical_solution_intro} also applies to subnormalized measures $p_i$, which then implies $(\lambda_1 p_1^\|,q_1) \leftrightarrow (\lambda_2 p_2^\|, q_2)$ (in particular $\lambda_1=\lambda_2=:\lambda$).
    The result then follows from \cref{thm:abscontintconv} below.
    \end{proof}
    \begin{lemma}\label{thm:abscontintconv} 
    Let $(X_i,\mu_i)$, $i=1,2$, be $\sigma$-finite measure spaces and let $p_i^\|,p_i^\perp,q_i\in L^1(X_i,\mu_i)$ be probability measures such that $p_i^\|\ll q_i$ and $p_i^\perp\perp q_i$.
    Suppose $(p_1^\|,q_1)\leftrightarrow (p_2^\|,q_2)$.
    Then for any $\lambda \in [0,1]$ we also have
    \begin{align}
        (\lambda p_1^\| + (1-\lambda) p_1^\perp,q_1) \leftrightarrow  (\lambda p_2^\| + (1-\lambda) p_2^\perp,q_2) 
    \end{align}
    \begin{proof}
        If $r\in L^1(X_1,\mu_1)$, we define $r_s$ and $r_c$ to be the singular and absolutely continuous parts of $r$ with respect to $q_1$.
        Note that $r\mapsto r_s$ and $r\mapsto r_c$ are linear substochastic maps and that $r=r_s+r_c$. We have $(p_1)_c=\lambda p^\|_1$, $(p_1)_s=(1-\lambda)p_1^\perp$, $(q_1)_s=0$ and $(q_1)_c=q_1$.
        Set $p_i=\lambda p_i^\|+(1-\lambda)p_i^\perp$.
        Let $T_1$ be a stochastic map taking $(p_1^\|,q_1)$ to $(p_2^\|,q_2)$.
        Define $$S_1(r) = r_s(X_1)p_2^\perp + T_1(r_c).$$
        Indeed, $S_1(p_1)= (1-\lambda)p_2^\perp+T_1(\lambda p_1^\|)=(1-\lambda)p_2^\perp+\lambda p_2^\|=p_2$ and $S_1(q_1)=T_1(q_1)=q_2$.
        Constructing $S_2$ mapping in the other direction analogously proves the claim.
    \end{proof}
\end{lemma}

We proved \cref{thm:classical_solution_without_ac} by reducing it to \cref{thm:classical_solution}.
Alternatively, we can adapt the proof of \cref{thm:classical_solution} so that it does not require absolute continuity to begin with. We sketch this here:

\begin{proof}[Direct proof of \cref{thm:classical_solution_without_ac}]
    Let $(p,q)$ be a dichotomy on a $\sigma$-finite measure space $(X,\mu)$.
    Let $p = p_c+p_s$ be the Lebesgue decomposition with respect to $q$.
    First, note that we may choose $\mu=p_s+q$ without loss of generality.
    Set $\overline \RR = \RR\cup\{+\oo,-\oo\}$.
    We start with an arbitrary dichotomy of probability measures on $L^1(X,\mu)$ and define a measurable function
    \begin{equation}
        f(x)=\infty\cdot \chi_{\supp(q) \setminus\supp(p)}(x) - \log\big(\tfrac{dp_c}{dq}(x)\big)\chi_{\supp(q)}(x) - \oo \cdot\chi_{\supp(p_s)}(x).
    \end{equation}
    This is just $f=-\log(\tfrac{dp}{dq})$ if we interpret $\tfrac{dp}{dq}$ as infinite on $\supp(p)\setminus\supp(q)$ (and as zero on $\supp(q)\setminus\supp(p)$ as usual).
    We again consider the push-forward measure $\tilde\mu := f_*(\mu)$ and $(\tilde p,\tilde q)=(f_*(p),f_*(q))$ which is a dichotomy in $L^1(\overline\RR,\tilde\mu)$.
    For measures $r\ll \tilde\mu$ we denote by $r_s$ (resp.\ $r_c$) the singular (resp.\ absolutely) part with respect to $\tilde q$.
    With this we have $\tilde\mu=\tilde p_s+\tilde q$ and it holds that
    \begin{equation}
        \tilde p(\{\oo\})=\tilde q(\{-\oo\})=0, \ \ \tilde p(\{-\oo\})= p_s(X), \ \ \tilde q(\{\oo\})=q(X\setminus\supp(p)).
    \end{equation}
    Just as in \cref{eq:Radon-Niko}, one sees that 
    \begin{equation}
        \frac{d\tilde p_c}{d\tilde q}(t) =e^{-t} \chi_{(-\oo,\oo]}(t).
    \end{equation}
    We now define stochastic maps $f_*:L^1(X,\mu)\to L^1(\overline{\RR},\tilde\mu)$ and $R:L^1(\overline{\RR},\tilde\mu)\to L^1(X,\mu)$.
    The map $f_*$ is the push-forward with $f$ and $R$ is given by
    \begin{equation}
        R(g\cdot\tilde\mu)=(g\circ f)\mu = g(-\oo) p_s + (g\circ f)\cdot q.
    \end{equation}
    This map is positive, preserves normalization and indeed satisfies $R(\tilde q)= R(\chi_{(-\oo,\oo]}\cdot\tilde\mu) = (\chi_{(-\oo,\oo]}\circ f)\cdot q=q$ and $R(\tilde p) = R\big((\chi_{\{-\oo\}}+ \tfrac{d\tilde p_c}{d\tilde q})\tilde\mu\big)= p_s + e^{-f}\cdot q = p_s+p_c=p$.
    With the definition \cref{def:better_def} we have
    \begin{equation}
        D_\alpha(p,q) = D_\alpha(\tilde p,\tilde q) = \int_{-\oo}^\oo e^{-t\alpha} d\tilde q(t),\quad \alpha\in(0,1).
    \end{equation}
    If we had two dichotomies $(p_i,q_i)$ with equal R\'enyi divergences then this implies that the two-sided Laplace transforms of $\tilde q_1$ and $\tilde q_2$ are equal.
    By injectivity of the Laplace transform (see \cref{thm:injectivity}), it follows that $\tilde q_1=\tilde q_2$ and hence that $(\tilde p_1)_c=e^{-t}\tilde q_1= e^{-t}\tilde q_2=(\tilde p_2)_c$.
    Furthermore, $(p_1)_s(X_1) = 1-(\tilde p_1)_c(\overline\RR) =1-(\tilde p_2)_c(\overline\RR) = (p_2)_s(X_2)$.
    Therefore, we also have $\tilde p_1=  (\tilde p_1)_c + (\tilde p_1)_s = (\tilde p_1)_c + (p_1)_s(X_1)\delta_{-\oo} = (\tilde p_2)_c + (p_2)_s(X_2)\delta_{-\oo} = \tilde p_2$.
    With this we also get $\tilde\mu_1 = \tilde\mu_2$ which proves interconvertibility because 
    $$
    (p_1,q_1)\leftrightarrow(\tilde p_1,\tilde q_1) = (\tilde p_2,\tilde q_2) \leftrightarrow (p_2,q_2).
    $$
\end{proof}

\section{Convertibility and interconvertibility for probability vectors}

\label{sec:explicit}
Here, we discuss the Koashi-Imoto minimal form and solution to the classical case (see \cref{thm:classical_solution_intro}) for pairs of probability vectors, where everything can be constructed explicitly.

\paragraph{Minimal form.}
Let $\mathbf p, \mathbf q\in \RR^n$ be probability vectors. Without loss of generality, we assume that $\mathbf p+\mathbf q$ has no zero entry and let $l$ be the number of zero entries of $\mathbf q$.
If necessary, apply a permutation to $\mathbf p$ and $\mathbf q$, so that
\begin{equation}
   \mathbf  r = (r_1,\dots,r_n) = \Big( \frac{p_1}{q_1},\frac{p_2}{q_2}, \ldots, \frac{p_n}{q_n}\Big) 
\end{equation}
is ordered such that $r_1\ge r_2\ge \ldots\ge r_n$. If $l>0$, $r_1=\ldots=r_l = \infty$. 
Denote by $m$ the number of \emph{distinct} entries of the vector $\mathbf r$. For example, if $\mathbf p = (\tfrac{1}{12},\tfrac12,\tfrac14,\tfrac{1}{12},\tfrac{1}{12},0)$ and $\mathbf q = (0,\tfrac{1}{12},\tfrac{1}{12},\tfrac16,\tfrac16,\tfrac12)$, then $\mathbf r=(\infty,6,3,\tfrac12,\tfrac12,0)$ and $m=5$.
Each entry of $\mathbf r$ appears with a certain multiplicity which we denote by $d_1,\ldots,d_m$ and we collect the corresponding indices in a set $I_j$ with size $d_j=|I_j|$ (in the above example this means $d_1=d_2=d_3=1,d_4=2,d_5=1$ and $I_1=\{1\}, I_2=\{2\}, I_3=\{3\},I_4=\{4,5\},I_5=\{6\}$).
The minimal form is given by the probability vectors $\tilde{\mathbf p}$ and $\tilde{\mathbf q}$ on $\RR^m$ with entries
\begin{equation}
    \tilde p_j = \sum_{k\in I_j} p_k,\quad \tilde q_j = \sum_{k\in I_j} q_k.
\end{equation}
Let us define $s(j) := \sum_{k=1}^{j-1}d_k$ and the vectors $\mathbf w^{(j)} \in \RR^{d_j}$ with
\begin{align}
    \quad w^{(1)}_k = \frac{p_{k}}{\tilde p_1},\quad w^{(j)}_{k} = \frac{q_{s(j)+k}}{\tilde q_j},\quad j=2,\ldots,m.
\end{align}
Note that $\frac{p_j}{q_j} = \frac{\tilde p_k}{\tilde q_k}$ for all $k$ and $j\in I_k$ as long as $0\neq \tilde q_k,\tilde p_k$.
The dichotomy $(\mathbf p,\mathbf q)$ can be interconverted with its minimal form $(\tilde{\mathbf p},\tilde{\mathbf q})$ using the stochastic matrices $T \in M_{m\times n}(\RR)$ and $R\in M_{n\times m}(\RR)$ given by
\begin{equation}\label{eq:explicit_matrices}
T=
\begin{pmatrix}%
    1\cdots1 &         &       &       \\
             & 1\cdots1&       &       \\
             &         & \ddots&       \\
    \undermat{d_1}{\phantom{1\cdots1}}   & \undermat{d_2}{\phantom{1\cdots1}}&       &\undermat{d_m}{1\cdots1}
\end{pmatrix},
\qquad
R=
\begin{pmatrix}
    w^{(1)}_1     &        & \\
    \vdots        &        & \\
    w^{(1)}_{d_1} & \ddots & \\
                  &        &w^{(m)}_1\\
                  &        & \vdots\\
                  &        &w^{(m)}_{d_m}
\end{pmatrix}.
\end{equation}
Both matrices $T$ and $R$ have a block form. In the case of $T$ these blocks have dimension $1\times d_i$ and for $R$ they have dimension $d_i\times 1$.
Note that $R$ could be written as $R=\mathrm{diag}(\mathbf w^{(1)},\ldots,\mathbf w^{(m)})$ with some abuse of notation.
We will now show that $(\tilde{\mathbf p},\tilde{\mathbf q})$ are the Koashi-Imoto minimal form of $(\mathbf p,\mathbf q)$. 
We may decompose $\RR^n$ and the probability vectors $\mathbf p$ and $\mathbf q$ as
\begin{equation}
    \RR^n = \bigoplus_{j=1}^m\RR^{d_j},\quad \mathbf{p} = \oplus_{j=1}^m \tilde p_j\, \mathbf{w_j},\quad \mathbf{q} = \oplus_{j=1}^m \tilde q_j \,\mathbf w_j.
\end{equation}
This is precisely the classical version of the Koashi-Imoto decomposition for the dichotomy $(\mathbf p,\mathbf q)$.
We see that the $\J_j$ spaces are all trivial (anything else would imply a form of non-commutativity) and that the $\K_j$ spaces are given by $\RR^{d_j}$ with the role of the states $\omega_j$ taken by the vectors $\mathbf w_j$.
If we "trace out" the $\mathbf w_j$'s, then we obtain the vectors $\tilde{\mathbf p}$ and $\tilde{\mathbf q}$.
That this decomposition is, in fact, the minimal one follows from the proof of \cref{thm:classical_solution_vecs}.
The stochastic matrix $T$ effectively implements the push-forward map $f_*$ with $f=-\log(\tfrac{dp}{dq})$ while $R$ implements the associated recovery map.
We will also explicitly see the minimality below.

\paragraph{Interconvertibility.}
Let $(\mathbf p_1,\mathbf q_1), (\mathbf p_2,\mathbf q_2)$ be classical dichotomies.
Without loss of generality we again assume the ordering as above.
Denote the minimal forms by $(\tilde{\mathbf p}_i, \tilde{\mathbf q}_i)$ and the stochastic matrices  interconverting them with $(\mathbf p_i,\mathbf q_i)$ by $T_i$ and $R_i$ (as in \cref{eq:explicit_matrices}).
In this case the interconvertibility, if possible, can always be achieved using the maps $R_1T_2$ and $R_2T_1$.
Collecting our results we get:

\begin{corollary}
    The following are equivalent
    \begin{enumerate}[(1)]
        \item $(\mathbf p_1,\mathbf q_1)$ and $(\mathbf p_2,\mathbf q_2)$ are interconvertible,
        \item\label{item:app_renyis} $D_\alpha(\mathbf p_1,\mathbf q_1)=D_\alpha(\mathbf p_2,\mathbf q_2)$ for all $\alpha$ in an open interval $(a,b)\subset (0,\infty)$  (with $(a,b)\subseteq (0,1)$ if $\mathbf p_i \ll \mathbf q_i$ does not hold.),
        \item\label{item:app_minimalvectors} $(\tilde{\mathbf p}_1,\tilde{\mathbf q}_1)=(\tilde{\mathbf p}_2,\tilde{\mathbf q}_2)$,
        \item $R_1 T_2 (\mathbf p_2,\mathbf q_2) = (\mathbf p_1,\mathbf q_1)$ and $(\mathbf p_2,\mathbf q_2) = R_2T_1(\mathbf p_2,\mathbf q_2)$.
    \end{enumerate}
\end{corollary}
In the case of finite-dimensional probability vectors the equivalence of \ref{item:app_renyis} and \ref{item:app_minimalvectors} can be seen explicitly as follows. We treat the general case where $\mathbf p\ll \mathbf q$ does not hold. 
Since $r_k$ is constant for  $k\in I_j$ we find (assuming $\tilde q_1=0$ and setting $\lambda = \sum_{j=2}^m \tilde p_j= 1-\tilde p_1$, compare wit \cref{def:better_def}) 
\begin{align}\label{eq:divs-vecs}
    \exp[(\alpha-1)D_\alpha(\mathbf p,\mathbf q)] = \lambda^\alpha \sum_{j=2}^m \sum_{k\in I_j} q_k r_k^\alpha = \lambda^\alpha \sum_{j=2}^m \tilde q_j \tilde r_j^\alpha = \exp[(\alpha-1)D_\alpha(\tilde{\mathbf p},\tilde{\mathbf q})]
\end{align}
with $\tilde r_j := \tilde p_j/\tilde q_j = r_k$ for $k\in I_j$. 
Therefore $D_\alpha(\mathbf p,\mathbf q) = D_\alpha(\tilde{\mathbf p},\tilde{\mathbf q})$. 
Since $\tilde r^\alpha_j = \exp(\alpha \log(\tilde r_j))$ and exponential functions $x\mapsto \exp(a x)$ are linearly independent for distinct $a$, we find that $D_\alpha(\mathbf p_1,\mathbf q_1) = D_\alpha(\mathbf p_2,\mathbf q_2) \ \ \forall \alpha\in(a,b)$ implies
\begin{align}
    \tilde q_{1,j} = \tilde q_{2,j}, \quad \tilde r_{1,j} = \tilde r_{2,j},\quad \lambda_1=\lambda_2, \quad j=2,\ldots,m
\end{align}
where we used $\lambda_i = \sum_{j=2}^m \tilde p_{i,j} = \sum_{j=2}^m \tilde r_{i,j} \tilde q_j$. Hence, we can reconstruct $\tilde p_{1,j} = \tilde p_{2,j}$ for $j=1,\ldots,m$. 
We thus find
\begin{align}
    D_\alpha(\mathbf p_1,\mathbf q_1) = D_\alpha(\mathbf p_2,\mathbf q_2) \ \ \forall \alpha\in(a,b)\quad \Leftrightarrow\quad (\tilde{\mathbf p}_1,\tilde{\mathbf q}_1) = (\tilde{\mathbf p}_2,\tilde{\mathbf q}_2). 
\end{align}

\paragraph{Lorenz curves.} It is useful to visualize the Koashi-Imoto decomposition for probability vectors. Given a dichotomy $(\mathbf p,\mathbf q)$ of $n$-dimensional probability vectors (again ordered according to the ratios $r_j$), consider the piece-wise linear function $x\mapsto L_{(\mathbf p,\mathbf q)}(x)$ whose graph is given by linearly joining the  $n+1$ points
\begin{align}
    (x_j,y_j) = (\sum_{k=1}^j q_k,\sum_{k=1}^j p_k),\quad j=1,\ldots,n,\quad (x_0,y_0)=0.
\end{align}
(The graph of) $L_{(\mathbf p,\mathbf q)}$ is called a \emph{Lorenz curve}.
Its slope between $x_j$ and $x_{j+1}$ is given precisely by $r_j$. 
Since $r_j \geq r_{j+1}$ the curve is concave,
see Fig.~\ref{fig:lorenz} for an example.
\begin{figure}
    \centering
    \includegraphics{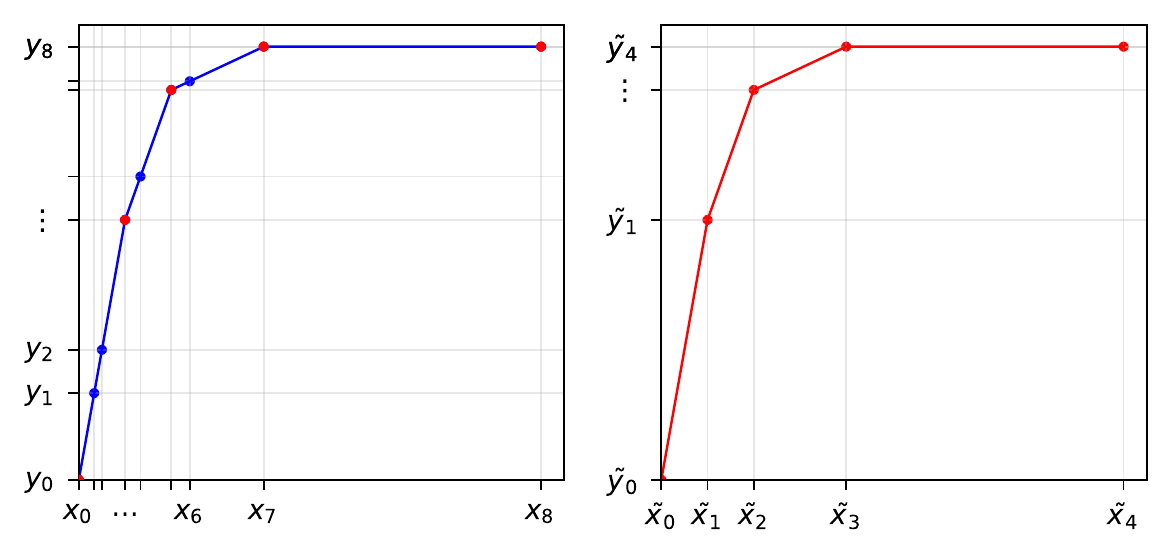}
    \caption{Koashi-Imoto minimal form from the point of view of Lorenz curves:
    The blue curve on the left corresponds to $L_{(\mathbf p,\mathbf q)}$. Removing all points from the Lorenz curve that lie in the interior of straight line segments (the blue points) yields the Lorenz curve $L_{(\tilde{\mathbf p},\tilde{\mathbf q})}$ on the right. This corresponds to the map $T$. Obviously, no further "anchoring points" can be removed from the curve without altering it.}
    \label{fig:lorenz}
\end{figure}
It is known that \cite{blackwell_equivalent_1953,ruch_mixing_1978,uhlmann_ordnungsstrukturen_1978,renes_relative_2016}
\begin{align}
    (\mathbf p_1, \mathbf q_1)\rightarrow (\mathbf p_2,\mathbf q_2) \Leftrightarrow L_{(\mathbf p_1, \mathbf q_1)}(x) \geq L_{(\mathbf p_2, \mathbf q_2)}(x)\quad \forall x\in[0,1].
\end{align}
If a Lorenz curve $L_{(\mathbf p, \mathbf q)}$
contains several consecutive linear segments with identical slope (corresponding to identical values of $r_j$), it is clear that the same curve can be obtained from a lower-dimensional vector, where the corresponding entries in $\mathbf p$ and $\mathbf q$ are simply summed up. 
Doing this for all consecutive linear segments with identical slopes precisely corresponds to the map $T$ above and yields $(\tilde{\mathbf p},\tilde{\mathbf q})$. From this picture it is also clear that the resulting dichotomy is minimal, see Fig.~\ref{fig:lorenz} for an illustration.

\paragraph{Convertibility from R\'enyi divergences for $\alpha>0$ alone.}
If $(\mathbf p_1,\mathbf q_1)\rightarrow (\mathbf p_2,\mathbf q_2)$ then $D_\alpha(\mathbf p_1,\mathbf q_1) \geq D_\alpha(\mathbf p_2,\mathbf q_2)$ due to the data-processing inequality.
Interestingly, the converse it not true: $D_\alpha(\mathbf p_1,\mathbf q_1) \geq D_\alpha(\mathbf p_2,\mathbf q_2)$ does in general not imply $(\mathbf p_1,\mathbf q_1)\rightarrow (\mathbf p_2,\mathbf q_2)$ \cite{brandao_second_2015}.
Still, we will now show that the knowledge of $D_\alpha(\mathbf p_1,\mathbf q_1)$ and $D_\alpha(\mathbf p_2,\mathbf q_2)$ for $\alpha\in(a,b)$ is sufficient to decide whether $(\mathbf p_1,\mathbf q_1)\rightarrow (\mathbf p_2,\mathbf q_2)$.

\begin{proposition}\label{prop:convertibility}
     There is an algorithm to determine if a given dichotomy $(\mathbf p_1,\mathbf q_1)$ is convertible to a dichotomy $(\mathbf p_2,\mathbf q_2)$, i.e., if $(\mathbf p_1,\mathbf q_1)\rightarrow(\mathbf p_2,\mathbf q_2)$, from the knowledge of the R\'enyi divergences $D_\alpha(\mathbf p_1,\mathbf q_1)$ and $D_\alpha(\mathbf p_2,\mathbf q_2)$ on any open subset $(a,b)\subset [0,\oo)$ on which they are both finite.
\end{proposition}

This follows from the following Lemma since knowing the Lorenz curves allows one to decide convertibility (see previous paragraph) and since we can reconstruct the Lorenz curves from $(\tilde{\mathbf p},\tilde{\mathbf q})$ algorithmically.

\begin{lemma}
    There is an algorithm to reconstruct $(\tilde{\mathbf p},\tilde{\mathbf q})$ from the knowledge of $D_\alpha(\mathbf p,\mathbf q)$ over any open subset $(a,b)\subset[0,\infty)$ on which $D_\alpha(\mathbf p,\mathbf q) <\infty$.
\end{lemma}

\begin{proof}
     We have seen in \eqref{eq:divs-vecs} that (assuming the general case with $\tilde q_1=0$ and $\lambda = 1 - \tilde p_1$)
    \begin{align}
        \exp((\alpha-1)D_\alpha(\mathbf p,\mathbf q)) = \lambda^\alpha \sum_{j=2}^m \tilde q_j \tilde r_j^\alpha = \mc L[t\mapsto \sum_{j=2}^m \tilde q_j \delta(t+\log(\lambda \tilde r_j)](\alpha),
    \end{align}
    where $\mc L$ is the (two-sided) Laplace-transform. Since the Laplace-transform is injective on any open interval (see \cref{thm:injectivity}), we can perform an inverse Laplace transform and read off $\tilde q_j$ and $\lambda \tilde r_j$ for $j=2,\ldots,m$. Since $\sum_{j=2}^m \tilde r_j \tilde q_j = \sum_{j=2}^m \tilde p_j = \lambda$, we obtain 
    \begin{align}
        \lambda = 1-\tilde p_1 = \sqrt{\sum_{j=2}^m \lambda \tilde r_j \tilde q_j}
    \end{align}
    and consequently $\tilde p_j$ for all $j$. 
    Thus we have reconstructed $\tilde{\mathbf p}$ and $\tilde{\mathbf q}$ from the knowledge of the R\'enyi divergences.
\end{proof}

\section{Proof of Theorem~\ref*{thm:interconvertibility} for infinite-dimensional systems}\label{sec:appendix}

\renewcommand{\A}{\mc M}

The direction \ref{it:normal_form} $\Rightarrow$ \ref{it:interconvertibility} is straightforward. 
For \ref{it:interconvertibility} $\Rightarrow$ \ref{it:normal_form}, consider the Koashi-Imoto minimal forms $\hat \rho_\theta\up i$ acting on $\hat\H_i =\bigoplus_j \J_j\up i$.
Define the von Neumann algebra $\A_i =\bigoplus_j \B(\J_j\up i)$ whose pre-dual $(\M_i)_*=\bigoplus_j \T(\J_j)$ contains $\hat\rho_\theta\up i$.
We have positive trace-preserving maps $\pi_i: \T(\H_{S_i})\to (\M_i)_*$, $\iota_i:(\M_i)_*\to\T(\H_{S_i})$ as in \cref{eq:normal_intercon}.
Let $T_1$ and $T_2$ be channels interconverting $\rho_\theta\up1$ and $\rho_\theta\up2$.
Then $\hat T_1 = \pi_2 T_1\iota_1$ and $\hat T_2=\pi_1T_2\iota_2$ are completeley positive trace-preserving maps between $(\M_1)_*$ and $(\M_2)_*$ which interconvert $\hat\rho_\theta\up 1$ and $\hat\rho_\theta\up2$.
By the Koashi-Imoto theorem, it now holds that $\hat T_2\hat T_1=\id$ and $\hat T_1\hat T_2=\id$.
Consider the unital completely positive maps $\alpha_i : \A_i\to \A_j$ dual to $\hat T_j$, $j\ne i$.
We then obtain $\alpha_1\alpha_2=\id_{\A_2}$ and $\alpha_2\alpha_1=\id_{\A_1}$.
Since the $\alpha_i$ are contractive, but their product is isometric, it follows that the $\alpha_i$ are both isometric.
From Kadison's theorem on completely positive isometries \cite[Thm.~7]{kadison}, we obtain that both $\alpha_i$ are $*$-isomorphisms and that $\alpha_1=\alpha_2^{-1}$.
In particular, $\A_1\cong\A_2$.
Consequently, the block structures are equivalent, i.e., $N_1=N_2$, and there is a reordering of indices such that the $j$th block of $\A_1$ is isomorphic to the $j$th block of $\A_2$.
The adjoints of the unitaries implementing $\alpha_1$ then map $\rho_{j|\theta}\up1$ to $\rho_{j|\theta}\up2$ and we get $p_{j|\theta}\up1=p_{j|\theta}\up2$.
    
It remains to be shown that interconverting channels are of the form \cref{eq:intercon_channels}.
From the Koashi-Imoto theorem we get that for all interconverting channels $R,T$ there are channels $E_j:\T(\K_j\up1)\to\T(\K_j\up1)$ such that $(R\circ T)|_{\T(\J_j\up1\ox\K_j\up1)}=\id_{\T(\J_j\up1)}\otimes E_j$ for all $j$.
We apply this to the specific channel $S =\oplus_j\big( (U_j(\placeholder)U_j^*)\ox(\omega_j\up 2\tr[\,\cdot\,])\big)$ which clearly maps $\rho_\theta\up1$ to $\rho_\theta\up2$.
Thus evaluation of $R\circ S$ on states of the form $\sigma=\oplus_j(U_j^*\sigma_jU_j\ox\omega_j\up1)$ shows 
$$ 
    R(\oplus_j(\sigma_j\ox\omega_j\up2))=R\circ S(\sigma)=\oplus_j(U_j^*\sigma_jU_j\ox\omega_j\up1).
$$
As this holds for all states $\sigma_j$, the channel acts as $U_j^*(\placeholder) U_j\otimes R_j$ on the $j$th block.
This then gives $R_j(\omega_j\up2)=\omega_j\up1$ as claimed.
Applying the analogous argument to $T$ finishes the proof.
\qed

\section{Quantum $f$-divergences}\label{sec:fdiv}

Following \cite{hiai}, this section is concerned with the notions of  \emph{standard}, \emph{maximal} and \emph{measured} quantum $f$-divergences.
\cite{hiai} treats the topic in the setting of general von Neumann algebras.
For an account of the development of the notions see there and the references therein.
Here, we are only concerned with classical and quantum states in the sense of the previous sections and will summarize the theory only for these cases.
The Petz quantum R\'enyi divergences are examples for \emph{standard} quantum $f$-divergences, while the maximal R\'enyi divergences can be obtained through \emph{maximal} quantum $f$-divergences (see below).
We will show that neither of the two families of quantum divergences is sufficient in the sense introduced above by demonstrating (well-known) constructions of classical dichotomies with equal divergences as an arbitrary but fixed (quantum) dichotomy.
By \cref{thm:commnogo}, interconvertibility cannot hold in general, and the claimed result follows.
We thus see, in particular, that neither Petz nor maximal R\'enyi divergences are sufficient.

Let $f:(0,\infty)\to \RR$ be a convex  function.
The (classical) $f$-divergence of two probability distributions $p\ll q$ is 
$$Q_f(p,q) = \int_X f\big(\tfrac{dp}{dq}\big)\, dq\in (-\infty, \infty].$$
We assume $f$ to be convex to ensure that the DPI holds for classical channels as well as to guarantee good behavior towards $0$ and $\infty$, see \cite{hiai}, but the definition in principle makes sense also for non-convex functions.
To define the quantum generalizations mentioned above, let $\rho = \sum_{i} p_i\ketbra{\psi_i}{\psi_i}$ be a density operator.
We denote its purification by $\ket{\rho} = \sum_i \sqrt{p_i} \psi_i\otimes \overline{\psi}_i$ in the Hilbert space $\H\ox\overline\H$.
Here, $\overline\H$ denotes the conjugate Hilbert space of $\H$.
We can identify linear operators $A$ on $\H$ with linear operators $\overline A$ on $\overline \H$.
Alternatively, we can identify $\ket\rho\equiv\sqrt{\rho}$ in the Hilbert space of Hilbert-Schmidt operators on $\H$ (and we will mainly employ this picture).

We start with the notion of standard quantum $f$-divergences, see \cite[Sec.~2]{hiai}.
Let $\sigma$ be another density operator with $\supp(\rho)\subseteq\supp(\sigma)$.
The \emph{relative modular operator} $\Delta_{\rho,\sigma}$ acts on a Hilbert-Schmidt operator $A$ by
$$\Delta_{\rho, \sigma}A = \rho A\sigma^{-1}$$
where $\sigma^{-1}$ again denotes the pseudo-inverse of $\sigma$.
On $\H\otimes\overline\H$ we identify $\Delta_{\rho, \sigma}\equiv \rho\otimes\overline\sigma^{-1}$.
$\Delta_{\rho, \sigma}$ is a self-adjoint operator which is unbounded if $\H$ is infinite dimensional.
For $f:(0,\infty)\to\RR$ convex, we define the \emph{standard quantum  $f$-divergence} of $\rho$ and $\sigma$ to be
$$Q_f^{\mathrm{std}}(\rho, \sigma) = \bra\sigma f(\Delta_{\rho, \sigma})\ket\sigma = \tr[\sqrt\sigma f(\Delta_{\rho, \sigma})(\sqrt\sigma)]\in (-\infty, \infty].$$
It is unitarily invariant, and if $f$ is operator convex then $Q_f^{\mathrm{std}}$ satisfies the DPI.
For $g(t) = t\log(t)$ we obtain the Umegaki relative entropy
$$Q_g^{\mathrm{std}}(\rho, \sigma) = D(\rho, \sigma) = \tr[\rho\log(\rho) - \rho\log(\sigma)].$$
Considering the family $f_\alpha:x\mapsto x^\alpha$ we recover \cite[Sec.~3.1]{hiai} the Petz quantum R\'enyi divergence as
$$
\petz\alpha(\rho, \sigma) = \tfrac 1 {\alpha-1} \log Q_{f_\alpha}^{\mathrm{std}}(\rho,\sigma).
$$
$f_\alpha$ is operator convex for $\alpha\in[1,2]$ and operator concave for $\alpha\in[0,1]$ \cite{Chansangiam2015}.
Therefore, $Q^{\mathrm{std}}_{f_\alpha}$ satisfies the DPI on $[1,2]$ while $-Q_{f_\alpha}^{\mathrm{std}}$ satisfies it on $[0,1]$.
By monotonicity of $\log$ and due to $\alpha-1$ being positive on $(1,2]$ and negative on $[0,1)$ we find that $\petz{\alpha}$ satisfies the DPI on $[0,2]$.

To construct the classical dichotomies let $\rho = \sum_i p_i R_i$ and $\sigma = \sum_j q_j S_j$ be the spectral decompositions of $\rho$ and $\sigma$.
Then $f(\Delta_{\rho,\sigma})(A) = \sum_{i,j} f(p_iq_j^{-1}) R_iAS_j$, so
$$Q_f^{\mathrm{std}}(\rho, \sigma) = \sum_{i,j}q_jf(p_iq_j^{-1}) \tr[R_iS_j].$$
Setting $P_{ij} = p_i \tr[R_iS_j]$ and $Q_{ij} = q_j \tr[R_iS_j]$ yields classical states with the same $f$-divergences as $\rho$ and $\sigma$ (for all $f$).
$P$ and $Q$ are known as the \emph{Nussbaum-Sko\l a distributions} \cite{NussbaumSzkola}.

We now turn towards the \emph{maximal quantum $f$-divergences}, see \cite[Sec.~4]{hiai}.
With $\rho$ and $\sigma$ as above and $T_{\rho, \sigma} = \sigma^{-\tfrac 12}\rho\sigma^{-\tfrac 12}\ge 0$ (possibly unbounded) they are given by the formula
$$Q_f^{\max}(\rho, \sigma) = \tr[\sigma f(T_{\rho, \sigma})]\in (-\infty, \infty].$$
It is again unitarily invariant and satisfies the data processing inequality provided that $f$ is operator convex.
The maximal quantum R\'enyi divergence is obtained from this as
$$\geom{\alpha}(\rho, \sigma) = \tfrac 1{\alpha-1}\log Q^{\max}_{f_\alpha}(\rho, \sigma).$$ 
As above for $\petz\alpha$ we find that $\geom\alpha$ satisfies the DPI on $[0,2]$.
It can be computed as
$$Q^{\max}_f(\rho, \sigma) = \min_{\substack{C,(p,q), T}} Q_f(p,q)$$
where the minimum ranges over all (possibly continuous) classical systems $C$, dichotomies with $p,q\in\states(C)$ and channels $T:\states(C)\to\states(S)$ that map $(p,q)$ into $(\rho, \sigma)$.
To show that the minimum is indeed attained, set $\omega = \rho + \sigma$ and consider $T_{\rho,\omega}$ with spectral decomposition $T_{\rho, \omega} = \int_{[0,1]} t\ d P_{\rho,\omega}(t)$.
Note that although $T_{\rho,\sigma}$ can be unbounded in general, $T_{\rho, \omega}$ is always bounded with $0\le T_{\rho, \omega}\le \1$.
This follows since
$S_\omega - T_{\rho, \omega} = T_{\sigma, \omega} \ge 0$ with $S_\omega$ denoting the projection onto $\supp(\omega) = \supp(\sigma)$.
In particular, the spectral measure of $T_{\rho, \omega}$ is indeed supported on $[0,1]$.
We can now define a finite Borel measure $\mu$ on $[0,1]$ by setting
$$\mu(A) = \tr[\omega P_{\rho, \omega}(A)] = \tr[E_{\rho, \omega}(A)] $$
where $E_{\rho, \omega}(A) := \sqrt{\omega}P_{\rho, \omega}(A)\sqrt{\omega}$.
We then define a channel $T:L^1(X,\mu)\to \T(\H)$ by setting
$$T(g) = \int g\ dE_{\rho,\omega}.$$
Let $p(t) = t$ and $q(t) = 1-t$, which are indeed probability densities relative to $\mu$.
It then holds that $T(1) = \omega$, $T(p) = \rho$ and $T(q) = \omega - \rho = \sigma$.
Furthermore, for every operator convex function $f$ it holds that $Q_f^{\max}(\rho, \sigma) = Q_f(p,q)$ \cite{hiai}.
Apart from operator convexity, nothing in this construction depends on the $f$ used to define the quantum $f$-divergence.
This proves the existence of a single pair of classical distributions with the same (maximal)  $f$-divergences as $\rho$ and $\sigma$.

Finally, to define the \emph{measured quantum Rényi divergence} we set, for any convex function,
$$Q^{\textup{meas}}_f(\rho, \sigma) = \sup_\M Q_f(\M(\rho), \M(\sigma))$$
where the supremum runs over all quantum-classical channels (i.e.\ measurements) $\M:\T(\H)\to L^1(X)$ (that is, the classical system is not fixed).
Note, that since the output system is classical, complete positivity is equivalent to positivity so that the data processing inequality for positive maps follows immediately from the definition.
The measured quantum Rényi divergences are now
\begin{equation}
    \meas \alpha (\rho, \sigma) =\frac1{\alpha-1} \log Q^{\textrm{meas}}_{f_\alpha}(\rho, \sigma)= \sup_\M D_\alpha(\M(\rho), \M(\sigma)).
\end{equation}
They satisfy the data processing inequality for all $\alpha\in (0, \infty)$ (since the output system is classical, we do not have the same restriction coming from operator convexity as for the Petz and maximal divergence).
Furthermore, for each $\alpha$ there exists a (projective) measurement realizing the supremum \cite[Ex.~5.16]{hiai} and $\meas\alpha$ is continuous in $\alpha\in(0,\infty)$ \cite{Mosonyi2022SomeCP}. In particular, the limit $\alpha\to 1$ yields the measured relative entropy $\meas{}$, cf.\  \cref{sec:zoo}.

\section{The insufficient zoo of divergences}
\label{sec:zoo}

In this appendix, we collect definitions and some relevant information regarding DPI for the most used families of quantum Rényi divergences. 
We include the respective limits $\alpha\to 1$ as part of the definition of the divergence.
This list is certainly not exhaustive, but we are not aware of more families being known.\footnote{There does exist the family of \emph{iterated mean divergences} introduced in \cite{BrownFawziFawzi} but since it is parametrized by a discrete parameter (that accumulates at $1$), it does not lie in our scope.}
The different properties are summarized in \cref{tab:renyis}.
Recall that we assume $\rho\ll \sigma$.

Every family that satisfies the DPI for positive maps clearly cannot be sufficient for detecting completely positive interconvertibility.
Apart from this, some families can be excluded completely from being sufficient, even for positive interconvertibility, since they allow for classical dichotomies with equal divergences as genuinely quantum ones.

\begingroup
\setlength{\tabcolsep}{10pt} 
\renewcommand{\arraystretch}{2} 
\begin{table}[h!]
    \centering
    \begin{tabular}{c|cccc}
        & cp-DPI & p-DPI & additive & $\alpha\to 1$\\
        \hline\hline%
        geometric/max $\geom\alpha$ & $[0,2]$ & $[0,2]$  & yes& $D^{\textrm{BS}}$\\
        sharp $D_\alpha^\#$ & $[1, \infty)$ & $[1, \infty)$ & subadditive & $D^{\textrm{BS}}$\\
        log-Euclidean  $D^\flat_\alpha$ &$[0,1]$& $[0,1]$ & yes & $D$\\
        Petz $D_\alpha$ & $[0,2]$ & $[0,1]$ & yes & $D$\\
        sandwiched/min  $\sand \alpha$ & $[\tfrac12,\oo)$ & $[\tfrac12,\oo)$ & yes & $D$\\
        $\alpha$-$z$ $D_{\alpha,z}$ & see \eqref{eq:cp-DPI_alpha-z} & $\big\{\substack{z\ge\max(\alpha,1-\alpha)\\0<\alpha<1}\big\}$ & yes & $D$\\
        measured $\meas \alpha$ & $(0,\oo)$ & $(0,\oo)$ & superadditive & $\meas{}$
\end{tabular}
    \caption{Summary of properties of quantum R\'enyi divergences. }
    \label{tab:renyis}
\end{table}
\endgroup

\subsection*{Geometric/Maximal}
The \emph{maximal quantum R\'enyi} divergence is defined as
\begin{equation}\label{eq:maxQRE}
    \geom\alpha(\rho,\sigma) = \frac1{\alpha-1} \log \tr\big[ \sigma \big(\sigma^{-\frac12}\rho\sigma^{-\frac12} \big)^\alpha \big].
\end{equation}
Its limit $\alpha\to 1$ is the Belavkin-Staszewski relative entropy $D^{\mathrm{BS}}(\rho,\sigma)$ \cite{belavkin_cast_1982}.
Additivity of $\geom\alpha$ is immediate from the definition and it satisfies the DPI for positive trace-preserving maps on $\Lambda=[0,2]$ (and is also maximal only on this region) \cite[Thm.~4.4]{hiai}, \cite[Sec.~4.2.3]{tomamichel}, see also \cref{sec:fdiv}.
This holds in finite as well as infinite dimensions.
Since it allows for classical dichotomies with equal divergences, see \cref{sec:fdiv}, this family cannot be sufficient.

\subsection*{Sharp}
This family has been introduced in \cite{Fawzi2021definingquantum}.
We briefly sketch some background to the definition:
For any operator monotone function $f:[0,\infty)\to[0,\infty)$ with $f(1)=1$ we set
\begin{equation}
    \sigma\#_f\rho = \sigma^{1/2}f(\sigma^{-1/2}\rho\sigma^{-1/2})\sigma^{1/2}    
\end{equation}
In the case of $f(x) = x^\alpha$, $\alpha\in(0,1)$, we write $\#_\alpha$.
Note, that in the notation of \cref{sec:fdiv} it holds that $\sigma\#_f\rho = \sigma^{1/2}f(T_{\rho, \sigma})\sigma^{1/2}$ so that
\begin{align}
    Q_f^{\max}(\rho, \sigma) &= \tr [\sigma\#_f\rho].
\end{align}
In particular, $\geom\alpha(\rho, \sigma) = \tfrac 1 {\alpha-1}\log\tr[{\sigma\#_\alpha\rho}]$.
With this one sets for $\alpha\in(1, \infty)$
\begin{equation}\begin{aligned}
    Q_\alpha^\#(\rho, \sigma) =& \inf\{\tr(A)\ |\ A\ge 0,\, \rho\le A\#_{1/\alpha}\sigma\}\\
    D_\alpha^\#(\rho, \sigma) =& \frac 1 {\alpha-1}\log Q_\alpha^\#.
\end{aligned}\end{equation}
Its limit $\alpha\to 1$ is the Belavkin-Staszewski relative entropy \cite{BerghEtAl2021}. It is subadditive and satisfies the DPI for all positive trace-preserving maps \cite[Prop.~3.2]{Fawzi2021definingquantum}.

$D^\sharp_\alpha$ can be seen as an example of a \emph{kringel divergence} introduced in \cite{BerghEtAl2021} that, for a given family $\DD_\alpha$ of quantum Rényi divergences, is defined as
$$\DD^\circ_\alpha(\rho, \sigma) = \inf_{A\ge \rho} \DD_\alpha(A, \sigma)$$
where $\alpha > 1$, while for $\alpha < 1$ the optimization evaluates to $-\infty$.
It is shown there that this family is again a family of Rényi divergences.
Being defined as an infimum, it is immediate that whenever $\DD_\alpha$ satisfies the DPI for (completely) positive maps, so does $\DD_\alpha^\circ$. 
To our best knowledge, these families have so far only been considered in finite dimension.

\subsection*{log-Euclidean}

This family is defined as
\begin{equation}
    D^\flat_\alpha(\rho, \sigma) = \frac 1{\alpha - 1} \log\tr\  e^{\alpha\log(\rho) + (1-\alpha)\log(\sigma)}.
\end{equation}
It essentially appears already in \cite{HIAI1993153} where it is related to $\geom\alpha$ although not directly in the context of quantum divergences.
See also \cite{MosonyiOgawa2017, mosonyi2023geometric}.
Its limit $\alpha\to 1$ is the quantum relative entropy $D$.
It is furthermore the $z\to\infty$ limit of the $\alpha$-$z$ family \cite{audenaert_-z-renyi_2015}
\begin{equation}
    D_\alpha^\flat(\rho,\sigma) = D_{\alpha, \infty} (\rho,\sigma)= \lim_{z\to\infty} D_{\alpha,z}(\rho,\sigma).
\end{equation}
It is thus immediate from the analogous statement for $D_{\alpha, z}$ that $D^\flat_\alpha$ is additive and satisfies the DPI for positive trace-preserving maps if $\alpha\in [0,1]$ (also in infinite dimensions).
For $\alpha > 1$ the limiting argument cannot be done and, in fact, $D^\flat$ doesn't satisfy the DPI in this region \cite[Lem.~3.17]{MosonyiOgawa2017}.

It is furthermore a special case of the recently introduced \emph{barycentric Rényi divergences} \cite{mosonyi2023geometric}.
Given two quantum divergences $D\up0$ and $D\up 1$ it is defined by setting
\begin{equation}
    Q^{\text{bary}}_\alpha (\rho,\sigma) = \sup_{0\le\tau\ll \rho} \tr\ \tau - \alpha D\up 0(\tau,\rho) - (1-\alpha)D\up 1(\tau,\sigma)
\end{equation}
(recall that by assumption $\rho\ll\sigma$) and then defining
\begin{equation}
    D^{\text{bary}}_\alpha (\rho,\sigma) = \frac{1}{\alpha-1}\log Q^{\text{bary}}_\alpha(\rho,\sigma).
\end{equation}
For $\alpha\in[0,1]$ they inherit validity of the DPI from the defining families while for $\alpha>1$ they don't necessarily define proper divergences if $D\up 0 \neq D\up 1$ since then $D^{\text{bary}}_\alpha(\rho,\rho) >0$ is possible.
$D^\flat$ is obtained by choosing $D\up 0 = D\up 1 = D$ to be the quantum relative entropy.

\subsection*{Petz}
The \emph{Petz quantum Rényi divergence} is defined as 
\begin{equation}
    \petz{\alpha}(\rho, \sigma) = \frac 1 {\alpha-1} \log\tr[\rho^\alpha\sigma^{1-\alpha}],
\end{equation}
where $\sigma^{1-\alpha}$ denotes the \emph{pseudo-inverse} of $\sigma$. The limit $\alpha\to 1$ is the quantum relative entropy $D(\rho,\sigma)$.
Additivity of $\petz{\alpha}$ is immediate from the definition, and it satisfies the DPI for completely positive trace-preserving maps on $\Lambda=[0,2]$ \cite[Sec.~4.4]{tomamichel}, see also \cref{sec:fdiv}.
For $\alpha\in[0,1]$, it is furthermore known to satisfy DPI for positive trace-preserving maps (see $\alpha$-$z$ divergence below).
Both is true in infinite dimensions.
Since they allow for classical dichotomies with equal divergences, see \cref{sec:fdiv}, this family cannot be sufficient.

\subsection*{Sandwiched/Minimal}
The \emph{minimal quantum R\'enyi divergence} (or \emph{sandwiched R\'enyi divergence}) is given by
\begin{equation}\label{eq:minQRE}
    \sand \alpha(\rho,\sigma) = \frac1{\alpha-1}\log \tr\big[ \big(\sigma^{\frac{1-\alpha}{2\alpha} }\rho \sigma^{\frac{1-\alpha}{2\alpha}}\big)^\alpha\big] =: \frac{1}{\alpha-1}\log Q^{\min}_\alpha(\rho,\sigma).
\end{equation}
It satisfies the DPI for positive trace-preserving maps, also in infinite dimension, on $\Lambda = [\tfrac12,\oo)$ \cite[Sec.~4.3]{tomamichel}, \cite{berta_variational_2017} and \cite[Thm.~3.16]{hiai}.
The limit $\alpha\to 1$ is the quantum relative entropy.

\subsection*{$\alpha$-$z$}
A class of R\'enyi divergences that generalizes both the minimal quantum R\'enyi divergence as well as the Petz quantum R\'enyi divergence is the two-parameter \emph{$\alpha$-$z$ quantum R\'enyi divergence} \cite{audenaert_-z-renyi_2015}, defined as
\begin{align}
    D_{\alpha,z}(\rho,\sigma) = \frac{1}{\alpha-1}\log \tr\big[ \big(\sigma^{\frac{1-\alpha}{2z} }\rho^{\frac{\alpha}{z}} \sigma^{\frac{1-\alpha}{2z}}\big)^z\big].
\end{align}
For $z=\alpha$, it coincides with the minimal quantum R\'enyi divergence, while for $z=1$, it coincides with the Petz quantum R\'enyi divergence and thus satisfies DPI in these cases for the same range of $\alpha$.
Its limit $\alpha\to 1$ is the quantum relative entropy \cite{Mosonyi2022SomeCP}.
DPI for completely positive trace-preserving maps in finite dimension is known to hold \emph{if and only if} one of the following cases applies \cite{Zhang2020EqualityCO}:
\begin{equation}\label{eq:cp-DPI_alpha-z}
    \begin{cases}
        z\ge\max(\alpha,1-\alpha),&\text{if\ } \alpha\in(0,1),\\ \tfrac \alpha 2 \le z \le \alpha &\text{if\ }\alpha\in(1,2], \\ \alpha - 1 \le z\le \alpha,&\text{if\ } \alpha\ge2.
    \end{cases}
\end{equation}
If both $\alpha\in (0,1)$ and $z\ge \max(\alpha, 1-\alpha)$ are satisfied it is furthermore known to satisfy DPI also for positive trace-preserving maps even in infinite dimensions \cite{kato2023alphazrenyi}.
It is clearly additive for all $\alpha$ and $z$, see also \cite{Mosonyi2023alphaz, kato2023alphazrenyi}.
Since the minimal quantum R\'enyi divergence does not admit a classical simulation, neither does the $\alpha$-$z$ quantum R\'enyi divergence for general values of $(\alpha,z)$. 
Furthermore, since the Araki-Lieb-Thirring inequality (see for example \cite[Thm.~4.3]{sutter_approximate_2018}) implies that $D_{\alpha,z}(\rho,\sigma)$ is non-constant in $z$ if and only if $\rho$ and $\sigma$ commute, there cannot exist a classical simulation working for a fixed value of $\alpha$ and more than one value of $z$.

\subsection*{Measured}
The \emph{measured quantum Rényi divergence} is given as
\begin{equation}
    \meas \alpha (\rho, \sigma) = \sup_\M D_\alpha(\M(\rho), \M(\sigma))
\end{equation}
where the supremum runs over all measurements $\M$, i.e., (completely) positive maps with abelian range.
Here, $Q_f$ denotes the classical $f$-divergence, see \cref{sec:fdiv}.
Given any positive, trace-preserving map $T$, the map $\M\circ T$ is a measurement, too. Hence $\meas\alpha$ satisfies DPI for positive trace-preserving maps for each value of $\alpha$ in finite as well as infinite dimension.
From the DPI we get that $\meas\alpha\le\sand\alpha$ for $\alpha\geq 1/2$. The existence of a pair of states for which the inequality is strict together with $\sand\alpha$ being the minimal additive divergence satisfying DPI \cite{tomamichel} shows that $\meas\alpha$ is not additive. However, clearly $\meas\alpha$ is \emph{super-additive}:
\begin{align}
    \meas\alpha(\rho_1\ox\rho_2,\sigma_1\ox\sigma_2) &= \sup_{\M} D_\alpha(\M(\rho_1,\ox\rho_2), \M(\sigma_1\ox\sigma_2)\\ 
    &\geq \sup_{\M_1,\M_2} D_\alpha(\M_1(\rho_1)\ox\M_2(\rho_2),\M_1(\sigma_1)\ox\M_2(\sigma_2))\\ 
    &= \meas\alpha(\rho_1,\sigma_1) + \meas\alpha(\rho_2,\sigma_2).
\end{align}
As far as we are aware, the measured R\'enyi divergence does not allow for a classical dichotomy with equal R\'enyi divergence for all $\alpha$, since the optimal measurement $\M$  depends on $\alpha$ in general.
For $\alpha=1$ we obtain the \emph{measured relative entropy} $\meas{}(\rho,\sigma)$ defined as
\begin{align}
\meas{}(\rho,\sigma) = \sup_{\M} D(\M(\rho),\M(\sigma)) = \sup_{\M}\lim_{\alpha\to 1}D_\alpha(\M(\rho)),\M(\sigma)),
\end{align}
which coincides with the quantum relative entropy if and only if $[\rho,\sigma]=0$ \cite[Prop. 2.36]{sutter_approximate_2018}.

\printbibliography
\end{document}